\documentclass[10pt,journal,compsoc]{IEEEtran}
\IEEEoverridecommandlockouts
\usepackage{amsmath}
\usepackage{graphicx}
\usepackage{xcolor}

\usepackage{algorithm, algorithmic,xpatch}
\usepackage[algo2e,algosection, boxruled, linesnumbered]{algorithm2e} 

\usepackage{mathtools}
\usepackage{minipage-marginpar}

\usepackage{paralist}

\usepackage{subcaption}

\usepackage{amsthm}

\usepackage{booktabs} 
\usepackage{minipage-marginpar}
\usepackage{booktabs}
\usepackage{multirow}
\usepackage{float}
\usepackage{mathtools}
\usepackage[thinc]{esdiff}

\usepackage{amsmath}
\usepackage{amsfonts}
\usepackage{stmaryrd}

\makeatletter
\newcommand{\xMapsto}[2][]{\ext@arrow 0599{\Mapstofill@}{#1}{#2}}
\def\Mapstofill@{\arrowfill@{\Mapstochar\Relbar}\Relbar\Rightarrow}
\makeatother

\usepackage[noadjust]{cite}

\newtheorem{definition}{Definition}
\newtheorem{example}{Example}
\newtheorem{property}{Property}
\newtheorem{theorem}{Theorem}
\newtheorem{assumption}{Assumption}

\DeclareCaptionFormat{algor}{%
  \hrulefill\par\offinterlineskip\vskip1pt%
  \textbf{#1#2}#3\offinterlineskip\hrulefill}
\DeclareCaptionStyle{algori}{singlelinecheck=off,format=algor,labelsep=space}
\newcounter{parentalgorithm}

\makeatletter

\newenvironment{subalgorithms}{%

  \refstepcounter{algorithm}%

  \protected@edef\theparentalgorithm{\thealgorithm}%
  \setcounter{parentalgorithm}{\value{algorithm}}%

  \setcounter{algorithm}{-1}%
  \def\thealgorithm{\theparentalgorithm\alph{algorithm}}%
  \ignorespaces
}{%
  \setcounter{algorithm}{\value{parentalgorithm}}%
  \ignorespacesafterend
}
\makeatother

\makeatletter
\xpatchcmd{\algorithmic}
  {\arabic{ALC@line}}
  {\theALC@line}
  {}{}
\xpatchcmd{\algorithmic}{1.2em}{1.5em}{}{}

\newcommand{\updatelinenoprint}{%
  \setcounter{parentcounter}{\value{ALC@line}}
  \setcounter{ALC@line}{0}
  \renewcommand{\theALC@line}{\theparentcounter.\alph{ALC@line}}
}

\renewcommand{\theALC@line}{\arabic{ALC@line}}
\newcounter{parentcounter}

\makeatother

\makeatletter
\newcommand*{\skipnumber}[2][1]{%
   {\renewcommand*{\theALC@line}[1]{}\STATE #2}%
   \addtocounter{ALC@line}{-#1}}
\makeatother

\ifCLASSINFOpdf
\else
\fi
\begin{document}
%
\title{Multi-Source Spatial Entity Linkage}
%
%
%
%

\author{Suela Isaj,
        Torben Bach Pedersen,~\IEEEmembership{Senior Member IEEE,}
        and~Esteban Zim\'{a}nyi
\IEEEcompsocitemizethanks{\IEEEcompsocthanksitem Suela Isaj and Torben Bach Pedersen are with the Dep. of Computer Science, Aalborg University, Denmark. \protect 
E-mail: \{suela,tbp\}@cs.aau.dk
\IEEEcompsocthanksitem Esteban Zim\'{a}nyi is with the Dep. of Computer and Decision Engineering, Universit\'{e} libre de Bruxelles, Belgium. \protect
E-mail: ezimanyi@ulb.ac.be}

}

%
%

\markboth{IEEE TRANSACTIONS ON KNOWLEDGE AND DATA ENGINEERING, VOL. X, NO. Y, MONTH 2019}%
{Shell \MakeLowercase{\textit{et al.}}: Bare Demo of IEEEtran.cls for Computer Society Journals}
%



\IEEEtitleabstractindextext{%
\begin{abstract}
Besides the traditional cartographic data sources, spatial information can also be derived from location-based sources. However, even though different location-based sources refer to the same physical world, each one has only partial coverage of the spatial entities, describe them with different attributes, and sometimes provide contradicting information. Hence, we introduce the spatial entity linkage problem, which finds which pairs of spatial entities belong to the same physical spatial entity. Our proposed solution (\emph{QuadSky}) starts with a time-efficient spatial blocking technique (\emph{QuadFlex}), compares pairwise the spatial entities in the same block, ranks the pairs using Pareto optimality
with the \emph{SkyRank} algorithm, and finally, classifies the pairs with our novel \emph{SkyEx-*} family of algorithms that yield 0.85 \emph{precision} and 0.85 \emph{recall} for a manually labeled dataset of 1,500 pairs and 0.87  \emph{precision} and 0.6  \emph{recall} for a semi-manually labeled dataset of 777,452 pairs. Moreover, we provide a theoretical guarantee and formalize the \emph{SkyEx-FES} algorithm that explores only 27\% of the skylines without any loss in \emph{F-measure}. Furthermore, our fully unsupervised algorithm \emph{SkyEx-D} approximates the optimal result with an \emph{F-measure} loss of just 0.01. Finally, \emph{QuadSky} provides the best trade-off between \emph{precision} and \emph{recall}, and the best \emph{F-measure} compared to the existing baselines and clustering techniques, and approximates the results of supervised learning solutions.
\end{abstract}

}

\maketitle

\IEEEdisplaynontitleabstractindextext

%
\IEEEpeerreviewmaketitle

\IEEEraisesectionheading{\section{Introduction}\label{sec:introduction}}

\IEEEPARstart{W}{eb} data and social networks are growing in terms of information volume and heterogeneity. Almost all online sources offer the possibility to introduce locations (geo-tagged entities accompanied by semantic details). A specific type of sources whose primary focus is locations is \emph{location-based sources}, such as Google Places, Yelp, Foursquare, etc. In contrast to cartographic data sources, locations in location-based sources have a hybrid form that stands between a \emph{spatial object} and an \emph{entity}. We refer to them as \emph{spatial entities} since they are spatially located but also identified by other attributes such as the name of the location, the address, keywords, etc. Spatial entities play a key role in several systems that rely on spatial information such as geo-recommender systems, selecting influential locations, search engines using geo-preferences, etc.

However, while a spatial object is identified only by the coordinates, this is not the case for spatial entities. Different spatial entities might co-exist in the same coordinates (shops in a shopping mall), or the same entity might be located in different but nearby coordinates across different sources (e.g., "Chicago Roasthouse" appears in Yelp and Google Places with coordinates 82 meters apart). The identity of a spatial entity is the combination of several attributes. Unfortunately, the identity of a spatial entity is sometimes difficult to infer due to the inconsistencies within and among the sources; each location-based source contains different attributes; some attributes might be missing and even contradicting. For example, source A contains the spatial entity "Lygten" in (57.436 10.534) with the keywords "coffee",  "tea", and "cocoa and spices", while source B contains "Restaurant Lygten" in (57.435 10.533) with the keyword "restaurant". We need a technique that can automatically decide whether these two spatial entities are the same real-world entity. The problem of finding which spatial entities belong to the same physical entity is referred to as \emph{spatial entity linkage} or \emph{spatial entity resolution}. We use the term entity linkage since we do not merge the entities \cite{brizan2006survey}.  

There are several works that apply entity linkage in various fields \cite{shu2017user,yui2011survey,firmani2016online,maskat2016pay,efremova2015multi,firmani2016online,maskat2016pay,edwards2016sampling,goga2013exploiting,panchenko2015large,isaj2019profile} but only little work on spatial entities \cite{karam2010integration,berjawi2014representing,morana2014geobench,Isaj2019multi}, even though they are central in geo-related research. The entities in the majority of the entity linkage research refer to people; thus, the methodologies and the models are based on the similarities that two records of the same individual would reveal. 
Moreover, these works do not address the spatial character of spatial entities. As for the works in spatial entity integration \cite{karam2010integration,berjawi2014representing,morana2014geobench}, their main contribution is a tool rather than an algorithm. What is more, the methods propose arbitrarily attribute weights and score functions without experimentation nor evaluation. In contrast to \cite{karam2010integration,berjawi2014representing,morana2014geobench}, the skyline-based algorithm (\emph{SkyEx}) proposed in \cite{isaj2019profile} is free of scoring functions and semi-arbitrary weights, and achieves good results. However, \emph{SkyEx} is dependent on a threshold number of skylines $k$, which can only be discovered through experiments, as the authors do not provide methods for estimating $k$. To sum up, \emph{on the one hand, there is a growing amount of information about spatial entities, both within a single source and across sources, which can improve the quality of the geo-information; on the other hand, the spatial entity linkage problem is hard to resolve not only because of the heterogeneity of the data but also because of the lack of appropriate and effective methods}.

In this paper, we address the problem of spatial entity linkage across different location-based sources. We significantly extend a previous conference paper \cite{Isaj2019multi}. As an overall solution building on \cite{Isaj2019multi}, first, we propose a method that uses the geo-coordinates to arrange the spatial entities into blocks. Then, we pairwise compare the attributes of the spatial entities. Later, we rank the pairs according to their similarities using our novel technique, \emph{SkyRank}. Finally, we introduce three approaches (\emph{SkyEx-F}, \emph{SkyEx-FES}, \emph{SkyEx-D}) for deciding whether the pairs of compared entities belong to the same physical entity. Our  contributions are:
(1) we introduce \emph{QuadSky}, a technique for linking spatial entities and we evaluate it on real-world data from four location-based sources; (2) we propose an algorithm called \emph{QuadFlex} that organizes the spatial entities into blocks based on their spatial proximity, maintaining the complexity of a quadtree and avoiding assigning nearby points into different blocks; (3) to rank the pairs by their similarity, we propose a flexible technique  (\emph{SkyRank}) that is based on the concept of Pareto optimality; (4) to label the pairs, we propose the \emph{SkyEx-*} family of algorithms that considers the ranking order of the pairs and fixes a cut-off level to separate the classes; (5) we introduce two threshold-based algorithms: \emph{SkyEx-F} that uses the \textit{F-measure} to separate the classes, and \emph{SkyEx-FES}, an optimized version of \emph{SkyEx-F}, which provides a theoretical guarantee to prune 73\% of the skyline explorations of \emph{SkyEx-F}; (6) we propose \emph{SkyEx-D}, a novel algorithm that is fully unsupervised and parameter-free to separate the classes.

Contributions 1 and 2 originate from \cite{Isaj2019multi}, contributions 5 and 6 are new, and 3 and 4 are significantly improved compared to \cite{Isaj2019multi}. The work in \cite{Isaj2019multi} reported very good results compared to the baselines, but had the following limitation: the proposed threshold-based labeling algorithm \emph{SkyEx} needed the threshold number of skylines $k$ as input, and there were no proposed solutions on how to fix  $k$, apart from experimenting with different values.
We address this limitation by first modifying the original \emph{SkyEx} in \cite{Isaj2019multi} as to only rank and not label the pairs, and we refer to it as \emph{SkyRank}. Then, we delegate the classification problem to three new algorithms, namely \emph{SkyEx-F}, \emph{SkyEx-FES} and \emph{SkyEx-D}. The experiments in \cite{Isaj2019multi} attempt to fix $k$ using \emph{precision}, \emph{recall} and \emph{F-measure}. We now formalize this rationale in our novel \emph{SkyEx-F} algorithm. We improve further by providing a theoretical guarantee that \emph{SkyEx-F} can be stopped before exploring the whole dataset, and propose the optimized \emph{SkyEx-FES} that prunes 80\% of the skyline explorations of \emph{SkyEx-F}. Furthermore, we introduce a novel approach for estimating the number of skylines (\emph{SkyEx-D}), which is fully unsupervised and parameter-free and closely approximates the threshold-based versions (\emph{SkyEx-F} and \emph{SkyEx-FES}). In the present paper, we provide a new set of experiments for \emph{SkyEx-FES} and  \emph{SkyEx-D}, and compare with \emph{SkyEx-F}, supervised learning and clustering techniques.

The remainder of the paper is structured as follows: first, we describe the state of the art in Sect. 2; then, we introduce our approach in Sect. 3; later, we detail the stages of our approach: the spatial blocking in Sect. 4, comparing the pairs in Sect. 5, ranking the pairs in Sect. 6, and estimating the $k^{th}$ level of skyline in Sect. 7; we analyze the complexity of our solution in Sect. 8; we provide experiments in Sect. 9; and finally, we conclude in Sect. 10.

\section{Related Work}
In this section, we describe some works on \emph{entity resolution}, \emph{spatial data integration}, and \emph{spatial entity linkage}.

\textbf{Entity resolution.} The \emph{entity resolution} problem has been referred in the literature with multiple terms including \emph{deduplication}, \emph{entity linkage}, and \emph{entity matching} \cite{dong2013big,firmani2016online}. Entity resolution has been used in various fields such as matching profiles in social networks \cite{shu2017user}, bioinformatics data \cite{yui2011survey},  biomedical data \cite{christen2004febrl}, publication data \cite{firmani2016online,maskat2016pay}, genealogical data \cite{efremova2015multi}, product data \cite{firmani2016online,maskat2016pay}, etc. The attributes of the entities are compared, and a similarity value is assigned. The decision of whether to link two entities or not is usually based on a scoring function. However, finding an appropriate similarity function that combines the similarities of attributes and decides on whether to link or not the entities is often difficult. Several works use a training set to learn a classifier \cite{edwards2016sampling,peled2013entity,goga2013exploiting}, others base the decision on a threshold derived through experiments \cite{panchenko2015large,quercini2017liaison}.  Other approaches decide the include the uncertainty of a match into the decision \cite{magnani2010survey}. Finally, matching the entities can also be based on the feedback of an oracle \cite{maskat2016pay,firmani2016online} or of a user \cite{maskat2016pay}.

\textbf{Spatial data integration.} There are several works on integrating purely spatial objects. Spatial objects differ from spatial entities mainly because a spatial object is fully determined by its coordinates or its spatial shape whereas a spatial entity, in addition to being geo-located, has a well-defined identity (name, phone, categories). The works on spatial object integration aim to create a unified spatial representation of the spatial objects from single/multiple sources. Schafers at al \cite{schafers2014simmatching} integrate road networks using rules for detect matching and non-matching roads based on the similarity in terms of the length, angles, shape, as well as the name of the street if available. The solutions in \cite{abdalla2016geospatial,tabarro2017webgis,balley2004modelling,walter1999matching} are purely spatial and discuss the integration of spatial objects originating from sensors and radars to have a better representation of the surface in 2D or even in 3D. These approaches cannot apply to spatial entities.

\textbf{Spatial entity linkage}. Accommodating the challenges of spatial entities for the entity resolution problem has been specifically addressed in \cite{karam2010integration,berjawi2014representing,morana2014geobench,sehgal2006entity,raimond2008data,Isaj2019multi}. The work in \cite{sehgal2006entity} is a bridge between the works in spatial data integration and spatial entity linkage because the entities have names, coordinates, and types but similarly to spatial objects, they refer to landscapes (rivers, deserts, mountains, etc.). The method used in  \cite{sehgal2006entity} is supervised and requires labeled data. Moreover, even the similarity of the attribute "type" is learned through a training set. Regarding \cite{karam2010integration,berjawi2014representing,morana2014geobench}, the main contribution of these works relies on designing a spatial entity matching tool rather than an integration algorithm. In \cite{morana2014geobench}, the spatial entities within a radius are compared with each other, and the value of the radius is fixed depending on the type of spatial entity. For example,  the radius is 50 m for restaurants and hotels, but 500 m for parks. All attributes (except coordinates) are compared using the Levenshtein similarity. Since the name, the geodata and the type of the entity are always present, they carry two-thirds of the weight in the scoring function whereas the weights of the website, the address and the phone number are tuned to one-third. The prototype of the spatial entity matching in \cite{berjawi2014representing} relies on a technique that arbitrarily uses an average of the similarity scores of all textual attributes without providing a discussing on this choice. Similarly to \cite{karam2010integration,berjawi2014representing}, the main contribution of the work in \cite{morana2014geobench} is designing a tool for spatial entity integration. The underlying algorithm considers spatial entities that are 5 m apart from each other and compares the name of the entities syntactically and the metadata related to an entity semantically. Finally, the decision is taken using the belief theory \cite{raimond2008data}. The works in \cite{karam2010integration,berjawi2014representing,morana2014geobench} lack an evaluation of the algorithms. The work in \cite{Isaj2019multi} proposes a scalable spatial quadtree-based blocking technique that not only fixes the distance between the spatial entities but also controls the density of the blocks. Then, the spatial entities of the same block are compared on their name (Levenshtein), address (custom) and categories (Wu\&Palmer using Wordnet). Finally, a threshold-based algorithm (\emph{SkyEx}) is used to separate the classes. However, instead of using fixed thresholds for each attribute similarity, \emph{SkyEx} abstracts the similarities into skylines and needs only one threshold number of skylines $k$ to separate the classes. The authors provide experiments and evaluations, nevertheless, they lack estimation techniques for fixing $k$. The present paper uses the solution in \cite{Isaj2019multi} for the spatial blocking and the pairwise comparisons. We use the skylines for the labeling process as in \emph{SkyEx}, but we propose three new algorithms (\emph{SkyEx-F}, \emph{SkyEx-FES} and \emph{SkyEx-D}) to separate the classes, fixing $k$ internally.

\textbf{Summary.} The general entity resolution approaches propose interesting solutions, but they do not consider the spatial character of a spatial entity. The majority are designed to match entities that represent individuals (profiles in social networks, authors and publications, medical records, genealogical connections, etc.) or even linking species in nature. The proposed solutions for entity resolution in individuals, either supervised or based on an experimental threshold, are learned on human entity datasets. One can not merely assume the resemblance of behaviors in a human entity dataset to a spatial entity one. The solutions in species in nature are based on domain-specific algorithms that have little to no applicability in other fields. There is little specific work in spatial entities \cite{karam2010integration,berjawi2014representing,morana2014geobench}, mostly focusing on a tool for spatial data integration rather than on the algorithm. 
In all these works, the scoring function is chosen arbitrarily and no evaluation provided.



%

\section{Spatial Entity Linkage}

In this section, we introduce the problem definition and our overall solution.
The basic concept used in this work is a \emph{spatial entity} such as places, businesses, etc. Spatial entities originate from location-based sources, e.g., directories with location information (yellow pages, Google Places, etc.)  and location-based social networks (Foursquare, Gowalla, etc.).

\begin{definition}
A spatial entity $s$ is an entity identified uniquely within a source $I$, located in a geographical point $p$ and accompanied by a set of attributes $A=\{a_i\}$.
\end{definition}

The attributes connected to $s$ can be categorized as: 
\emph{spatial}: the point where the entity is located, expressed in longitude and latitude; \emph{textual}: attributes that are in the form of text such as name, address, website, description, etc.; 
\emph{semantic}: attributes in the form of text that enrich the semantics behind a spatial entity, e.g., categories, keywords, metadata, etc.;
\emph{date, time or number}: other details about a spatial entity such as phone, opening hours, date of foundation, etc. An example of a spatial entity originating from Yelp can be a place named "Star Pizza" in the point (56.716 10.114), with the keywords "pizza, fast food", and with address "Storegade 31". The same spatial entity can be found again in Yelp or other sources, sometimes having the same attributes, more, less, or even attributes with contradictory values. Thus, there is a need for an approach that can unify the  information within and across different sources in an intelligent manner.

\emph{Problem definition: } \emph{Given a set of spatial entities $S$ originating from multiple sources, the spatial entity linkage problem aims to find those pairs of spatial entities  $\langle s_i,s_j \rangle$  that refer to the same physical spatial entity. }



\begin{figure*}[htb]
 \centering 
 
 \includegraphics[width=0.9\linewidth]{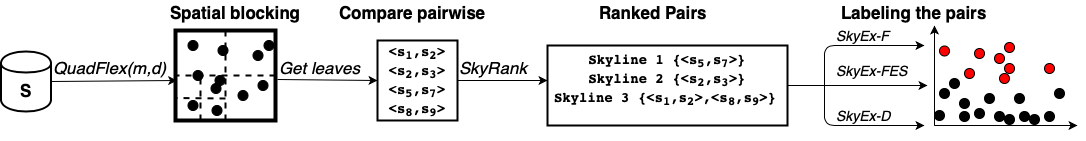}
    \caption{QuadSky approach} 
   \label{fig:quadsky}

 \end{figure*}
 
We propose \emph{QuadSky}, a solution based on a quadtree data partitioning and skyline exploration.  The overall approach is detailed in Fig.~\ref{fig:quadsky}. \emph{QuadSky} consists of four main parts: spatial blocking (\emph{QuadFlex}), pairwise comparisons, ranking the pairs (\emph{SkyRank}), and labelling the pairs (the \emph{SkyEx-*} family of algorithms). $S$ contains all spatial entities. We propose \emph{QuadFlex}, a quadtree-based solution that can perform the spatial blocking by respecting the distance between spatial entities and the density of the area. The output of  \emph{QuadFlex} is a list of leaves with spatial entities located nearby. Within the leaves, we perform the pairwise comparisons of the attributes. Then, we rank the compared pairs based on the skylines (concepts detailed in Sect. \ref{sub:skyex})   using the \emph{SkyRank} algorithm. In order to decide which pairs dictate a match and which not, we propose the \emph{SkyEx-*} family of algorithms (\emph{SkyEx-F}, \emph{SkyEx-FES}, and \emph{SkyEx-D}) that finds which skyline level best separates the pairs that refer to the same physical spatial entity (the positives class) from the rest (the negative class). In the following sections, we detail each of the phases of $QuadSky$. We use the notations in Table \ref{tab:notations} (We will explain them gradually during the paper). 

\begin{table}[t]
\centering
\scriptsize
\caption{Notations used throughout the paper}
\begin{tabular}{@{}ll@{}}
\toprule
\textbf{Notation}        & \textbf{Description}                                                        \\ \midrule
$s$                        & A spatial entity with a point $p$ and a set of attributes $\{a_i\}$            \\
$S$ & A set of spatial entities $\{s_i\}$ \\
$Q$ & A \emph{QuadFlex} structure used for spatial blocking \\
$P$                        & A set of pairs $\{\langle s_i, s_j \rangle\}$ \\
$\delta_a$ & The similarity of a pair in terms of attribute $a$ \\                           
 $u(\langle s_i,s_j\rangle)$ & The utility of a pair $\langle s_i,s_j\rangle$
\\
$Skyline(k)$ & A skyline of pairs $\{\langle s_i,s_j\rangle\}$ in the level $k$ 
\\

$K$ & The total number of skylines \\
$k$ &  A variable indicating the level of skyline
\\
$k_f$ &  A $k$ value fixed by \emph{SkyEx-F} and \emph{SkyEx-FES}
\\

$k_d$ & A $k$ value fixed by \emph{SkyEx-D}

\\

$P_k$ & Pairs of $P$ associated with a skyline \\

$P^+$ & A subset of pairs in $P$ classified as positive
\\

$P^-$ & A subset of pairs in $P$
classified as negative
\\

$F1(k)$ & The F-measure in the
$k^{\textit{th}}$ level of skyline
\\
$\mu_d$ & The mean of the distances between the two classes. 
\\
$\mu_d(k)$ & The function measuring $\mu_d$ in in each $k$ level of skylines
\\
 $\mu'_d(k)$ & The first derivative of $\mu_d(k)$

\\ \bottomrule
\end{tabular}
\label{tab:notations}
\end{table}

\section{Spatial Blocking}

\label{sub:spatialblocking}

Since spatial proximity is a strong indicator of finding a match, the first step is to group nearby spatial entities in blocks. Several generic blocking techniques have been discussed in \cite{papadakis2012beyond,papadakis2016comparative}, but mostly based on textual attributes and not applicable to spatial blocking. We propose a quadtree-based solution (\emph{QuadFlex}) that uses a tree data structure but also preserves the spatial proximity of spatial entities. 
A quadtree is a tree whose nodes are always recursively split into four children when the capacity is filled \cite{samet1984quadtree}. 
After the quadtree is constructed, the points that fall in the same leaf are nearby spatially. Hence, these leaves are good candidates to be spatial blocks. However, the existing quadtree algorithm needs to be adapted for spatial blocking. 

First, a quadtree needs a capacity (number of points) as a parameter. The capacity is not a meaningful parameter for spatial blocking, while the density of the area is a better candidate. For example, if the area is too dense (e.g., city center), even though the capacity is not reached, a further split would be more beneficial. 
On the contrary, two points in the countryside (e.g., a farm) might be farther apart, but they still might be the same entity. Second, a quadtree does not limit the distance between points. Even though two points might be in an area that respects the density, if they are quite distant from each other, it is not necessary to compare them. The maximal distance between two points in a child is the diagonal of the area (all quadtree children are rectangular). We used $m$, the diagonal of an area, as a parameter that controls the distance of points rather than comparing all distances between all spatial entities.
Finally, a quadtree splits into four children, and sometimes nearby points might fall into different leaves.  We modify the procedure of the assignment of the points into a child by allowing more than one assignment.

\begin{figure}[t]
 
    \centering
 \includegraphics[width=0.7\linewidth]{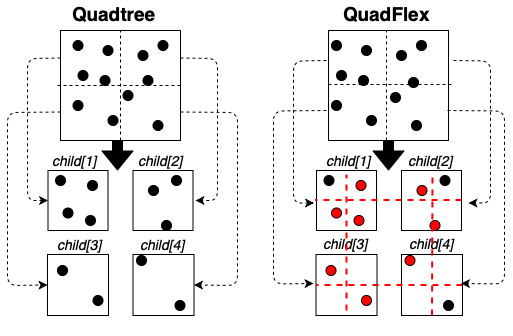}

  \caption{QuadFlex versus quadtree} 
  \label{fig:quads}

 \end{figure}

\begin{algorithm} [htbp]
\caption{QuadFlex algorithm} 
\label{alg:quadflex}
\footnotesize
\begin{algorithmic}[1]
\INPUT A set of entities  $S=\{s_i\}$, diagonal $m$, density $d$ \\ 
\OUTPUT The leaves QuadFlex  $Q$ $Q.\textit{leaves()}$ ;

\STATE Create $Q(m, d)$ where 
$Q$ has the dimensions of the bounding box of $S$

\FOR{\textbf{each} $s$ in $S$}
\STATE $Q.\textit{insert(s)}$ \COMMENT{Insert s into the QuadFlex}
\ENDFOR

\RETURN{\ $Q.\textit{leaves()}$}

\

\textbf{Method}\emph{ insert (s)} 
\IF{$\textit{this.children} \neq \empty$} 
\STATE $\textit{Indexes} \gets$ $\textit{getIndex(s)}$ \COMMENT{Find where $s$ belongs}
\FOR{\textbf{each} $i$ in $\textit{Indexes}$}
\STATE $\textit{this.child[i].insert(s)}$ \COMMENT{Insert $s$ to the children it belongs}
\ENDFOR
\ENDIF

\IF{$\textit{this.diagonal} > m$ $\OR$ $\textit{this.density} > d$} 
\STATE Split the current object $\textit{this}$ into 4 children 
\ENDIF
\STATE $\textit{Indexes} \gets$ $\textit{getIndex(s)}$ 
\FOR{\textbf{each} $i$ in $\textit{Indexes}$}
\STATE $\textit{this.child[i].insert(s)}$
\ENDFOR
\\
\RETURN

\

\textbf{Method}\emph{ getIndex (s)} 

\STATE Let $\textit{vertical-left}$  and $\textit{vertical-right}$  be the lines that pass at 0.25 and 0.75 of the width of $\textit{this}$, respectively
\STATE Let $\textit{horizontal-up}$  and $\textit{horizontal-down}$  be the lines that pass at 0.25 and 0.75 of the height of $\textit{this}$, respectively
\IF {$s$ is left of $\textit{vertical-right}$ and above $\textit{horizontal-down}$  }
\STATE $\textit{Indexes.add(1)}$  \COMMENT{ $s$ fits in  $\textit{child[1]}$}
\ENDIF
\IF{$s$ is right of $\textit{vertical-left}$ and above $\textit{horizontal-down}$ }
\STATE $\textit{Indexes.add(2)}$ \COMMENT{ $s$ fits in  $\textit{child[2]}$}
\ENDIF
\IF{$s$ is left of $\textit{vertical-right}$ and below $\textit{horizontal-up}$ }
\STATE $\textit{Indexes.add(3)}$
\COMMENT{ $s$ fits in  $\textit{child[3]}$}
\ENDIF
\IF{$s$ is right of $\textit{vertical-left}$ and below $\textit{horizontal-up}$ }
\STATE $\textit{Indexes.add(4)}$ \COMMENT{ $s$ fits in  $\textit{child[4]}$}
\ENDIF
\\
\RETURN{ $\textit{Indexes}$}

\end{algorithmic}
\end{algorithm} 
Fig.~\ref{fig:quads} shows the modifications that we do to the construction of the traditional quadtree for our version \emph{QuadFlex}. The traditional quadtree divides the area of each parent into four smaller areas, the children. A point belongs only to one child. In our modification, the area will split into 4 children in the same way as a quadtree (at 0.5 of the height and 0.5 of the width of the parent), but when we assign a point to a child, we will consider including points that fall shortly outside the border in the current child, too. For example, in Fig. 2, \emph{QuadFlex} physically splits in the same way as the quadtree, but the red dashed line shows the area that will be considered for including neighboring points. The red points are in the overlapping regions and will be included in more than one child.  
Algorithm \ref{alg:quadflex} details the procedure for retrieving the spatial blocks with \emph{QuadFlex}. The algorithm creates the root of the \emph{QuadFlex} tree with the bounding box of the data and parameters $m$ and $d$ (line~1). Then, it inserts each spatial entity into the \emph{QuadFlex} (line~3) and finally returns its leaves. The methods $insert(s)$ and $getIndex(s)$ are self calls on the \emph{QuadFlex} object ($\textit{this}$). The insertion procedure is similar to the traditional quadtree except that the constraint is not the capacity but the diagonal of the area $m$ (maximal distance between points) and the density of the area $d$. Hence, if the diagonal of the \emph{QuadFlex} is more than the distance $m$ or the density is larger than our defined value $d$ (line~12), the \emph{QuadFlex}, similarly to a quadtree, will split into four children. However, in contrast to the traditional quadtree, a spatial entity might belong to more than one child. The method $getIndex(s)$ gets the list of indexes of the children where the new point will be assigned. 
\emph{Even though $Q$ splits into 4 children in the same way as a quadtree, the lines $\textit{vertical-left}$, $\textit{vertical-right}$, $\textit{horizontal-up}$,  and $\textit{horizontal-down}$ allow a logical overlap of the areas and thus, neighboring spatial entities will not be separated}.

\section{Pairwise Comparisons}
\label{sub:pairwisecomp}

After the spatial blocking, we perform a pairwise comparison of spatial entities that fall in the same leaf. Next, we describe the metrics for different types of attributes.

\textbf{Textual Similarity. }
We measure the textual similarity of spatial entities using the edit distance between the words. The Levenshtein distance \cite{levenshtein1966binary} between string $s_1$ and string $s_2$ $d(s_1,s_2)$ is the number of  edits (insertion, deletion, change of characters) needed to convert string $s_1$ to string $s_2$. We define the similarity as:
\begin{equation}
\textit{TextSim} (s_1,s_2) = 
  (1 -  \frac{d(s_1,s_2)}{\textit{max}(|s_1|, |s_2|)}) 
\end{equation}
\begin{example}
Let us consider "Skippers Grill" and "Skippers Grillbar". The Levenshtein distance to convert "Skippers Grill" to "Skippers Grillbar" is 3 (3 insertions). The lengths of the first and the second string are 14 and 17 respectively. So,  $\textit{TextSim}(\text{"Skippers Grill"}$, $\text{"Skippers Grillbar"})$ $= 1-(3/\textit{max}(14,17) = 0.8235$. 
\end{example}

Note here that not all textual attributes can be handled similarly. String similarity metrics are usually appropriate for attributes like names, usernames, etc. Some other textual attributes require other metrics that need to be customized. In this paper, we consider the address as a specific textual attribute. The similarity between two addresses cannot be measured with Levenshtein, Jaccard, Cosine, etc. since a small change in the address might be a giant gap in the spatial distance between the entities. For example, "Jyllandsgade 15 9480 L{\o}kken" and "Jyllandsgade 75 9480 L{\o}kken" have a distance of 1 and Levenshtein similarity of 0.963, but they are 650 meters apart. However, "Jyllandsgade 15 9480 L{\o}kken" and "Jyllandsgade 15 9480 L{\o}kken Denmark" have a distance of 8 and Levenshtein similarity of 0.772, but they are the same building. In \cite{karam2010integration,berjawi2014representing} the address is considered as another textual attribute. In our case, we perform some data cleaning (removing commas, punctuation marks, lowercase, etc.), and then we search for equality or inclusion of the strings. We assign a similarity of 1.0 in the case of equality, 0.9 in the case of inclusion, and 0.0 otherwise.

\textbf{Semantic Similarity.} 
The similarity of fields like categories, keywords, or metadata cannot be compared only syntactically. Sometimes, several synonyms are used to express the same idea. Thus, we need to find a similarity than considers the synonyms as well. We use Wordnet \cite{fellbaum2010wordnet} for detecting the type of relationship between two words and Wu\& Palmer similarity measure ($\textit{wup}$) \cite{wu1994verbs}. 
The semantic similarity between two spatial entities is the maximal similarity between their list of categories, keywords, or metadata. The semantic similarity of the spatial entities $s_1$ and $s_2$ is:
\begin{equation}
    \textit{SemSim}(s1,s2) = \textit{max} \{\textit{wup}(c_i,c_j)\}
\end{equation}
where $c_i \in C_1$ and $c_j \in C_2$ and $C_1$ is the set of keywords of $s_1$ and $C_2$ is the set of keywords $s_2$.

\begin{example}
Let us take an example of two spatial entities $s_1$ and $s_2$ and their corresponding semantic information expressed as keywords $C_1=\{\text{"restaurant"}, \text{"italian"}\}$ and $C_1=\{\text{"food"}, \text{"pizza"}\}$. The similarity between each pair is $\textit{wup}(\text{"restaurant"},\text{"food"})=0.4$, $\textit{wup}(\text{"italian"}, \text{"food"})=0.4286$, $\textit{wup}$$(\text{"restaurant"}, \text{"pizza"})=0.3333$ and $\textit{wup}$$(\text{"italian"}, \text{"pizza"})=0.3529$. Finally, the semantic similarity of $s_1$ and $s_2$ is $ \textit{SemSim}(s1,s2) = \textit{max}\{0.4$, $0.4286$, $0.3333$, $0.3529\}=0.4286$.
\end{example}

\textbf{Date, Time, or Numeric Similarity. } The similarity between two fields expressed as numbers, dates, times or intervals is a boolean decision (true or false). 
Even though the similarity of these fields relies only on an equality check, most of the effort is put in data preparation. For example, the different phone formats should be identified and cleaned from prefixes. Other data formats like intervals (opening hours) might require temporal queries for similarity, inclusion, and intersection of the intervals. In this paper, we do not compute the similarity between these attributes as we use them to construct the ground truth.


\section{Ranking the Pairs}

\label{sub:skyex}

After the pairwise comparison, the pairs have $n$ similarity values, one for each attribute. We denote as $\delta_a$ the similarity of two spatial entities for attribute $a$. For example, a pair $\langle s_1, s_2 \rangle$  is represented as  $\{\delta_{a_1},..., \delta_{a_n}\}$. The problem that we need to solve is which $\langle s_i, s_j \rangle$ pairs indicate a strong similarity to be considered for a match. The related work solutions propose using a classifier \cite{edwards2016sampling,goga2013exploiting,kopcke2010evaluation} or experimenting with different thresholds \cite{kopcke2010evaluation,panchenko2015large,quercini2017liaison}.
We propose a more relaxed technique that uses Pareto optimality  \cite{censor1977pareto} for filtering the positive class. \emph{ A solution  $(x,y)$ is Pareto optimal when no other solution can increase $x$ without decreasing $y$}. The points in the same Pareto frontier or skyline have the same utility.  Widely used in economics and multi-objective problems,  Pareto optimality is free of weights and similarity score functions.  In the context of entity resolution, the skylines provide a selection of points that are better than others, but without quantifying how much better. The pairs that refer to the same physical spatial entity (the positive class) are expected to have high values of $\delta$, and consequently, form the first skylines.  
Under the assumption that the best values of $\delta$ belong to the pairs from the positive class, we label the pairs up to the $k^{th}$ skyline as the positive class and the rest as the negative. \emph{To the best of our knowledge, we are the first to propose a Pareto optimal solution for detecting matches for an entity linkage problem}.

\begin{definition}
An attribute $a$ is positive discriminating if its similarity $\delta_a$ indicates a positive class rather than a negative. 
\end{definition}

An example of a positive discriminating attribute is the similarity of name. A higher name similarity is more likely to indicate a match than a non-match. 
For example, the name similarity for \emph{Mand \& Bil} and \emph{Mand og Bil} is 0.75, and for \emph{Solid} and \emph{Sirculus ApS} is 0.16 . Hence, the former pair has a higher probability of being a match than the second. Examples of negative discriminating attributes are the edit distance between two names. If the distance between the names is high, then the pairs are less likely to be a match.

\begin{definition}
The utility of a positive discriminating attribute $a$, denoted as $u_a$, is the contribution of the attribute similarity $\delta_a$ to reveal a match, using Pareto Optimality ($\delta_a  \xmapsto{\textit{Pareto Optimality}} u_a$). 
\end{definition}

Each attribute similarity contributes to the labeling problem. Intuitively, a higher similarity $\delta_a$ of $a$ has a higher utility than a lower value of $\delta_a$. Hence, if $\delta_a(\langle s_1,s_2\rangle) > \delta_a(\langle s_3,s_4\rangle) $, then $u_a(\langle s_1,s_2\rangle) > u_a(\langle s_3,s_4\rangle)$.

\begin{definition}
The utility of a pair denoted as $u(\langle s_i,s_j\rangle)$ is sum of the utilities of each of the attributes.  
$u(\langle s_i,s_j\rangle) = \sum_{i=1}^n u_{a_i}$.
\end{definition}

Note that the utility of a pair is not the sum of the similarities of the attributes ($u(\langle s_i,s_j\rangle) \neq \sum_{i=1}^n \delta_{a_i}$) but the sum of their utilities ($u(\langle s_i,s_j\rangle) = \sum_{i=1}^n u_{a_i}$). Nevertheless,  $u(\langle s_i,s_j\rangle) = \sum_{i=1}^n \delta_{a_i} =  \sum_{i=1}^n u_{a_i}$ is a specific case.

\begin{definition}
A skyline of level $k$, $\textit{Skyline}(k)$, is the collection of  pairs  $\langle s_i,s_j\rangle$ of equal utility such that $u_{\textit{Skyline}(k)} > u_{\textit{Skyline}(k+1)}$.
\label{def:skyline}
\end{definition}

Obviously, $\textit{Skyline}(1)$ is the Pareto optimal frontier with the best values of $\delta_a$. In order to continue with $\textit{Skyline}(2)$, the points of   $\textit{Skyline}(1)$ are removed, and the frontier is calculated again. Every time we explore level $k$, the values in $\textit{Skyline}(k)$ are the ones with the highest utility. This means that \emph{there is no other point in a lower level that can bring a higher utility to the positive class}. This procedure continues until all the pairs are ranked according to their skyline. Algorithm \ref{alg:pareto} formalizes our proposed procedure Skyline Ranking (\emph{SkyRank}) for ranking the pairs. The input is the set of pairs $P$ produced from the \emph{QuadFlex} blocking technique and the number of skyline levels $k$ that we will explore. We find the points with the best combinations of $\delta$ that dominate the rest of the points and, consequently, have a higher utility (line~3). Then, we put these points in $P_k$, which keeps the explored skylines and remove them from $P$ (line~5).  We stop when all the pairs are assigned to a skyline. 
\begin{algorithm} [htbp]
\caption{Skyline Ranking (SkyRank) } 
\label{alg:pareto} 
\footnotesize
\begin{algorithmic}[1]
\INPUT A set of pairs  $P=\{\langle s_i, s_j \rangle \}$ \\ 
\OUTPUT A set of pairs and their skyline $P_k=\{\langle s_i, s_j \rangle, k \}$ ;
\STATE $P_k \gets \emptyset$
\WHILE{ $\lvert P_k \rvert < \lvert P \rvert$ }
\STATE Filter  $\textit{Skyline}(k)=\{ \langle s_i, s_j \rangle \}$  $\lvert$ $\forall \langle s', s'' \rangle \in P - \{ \langle s_i, s_j \rangle \} $ , $ u(\langle s_i, s_j \rangle) > u \langle s', s'' \rangle \}$  \COMMENT{Find the Skyline }

\STATE Add $\textit{Skyline}(k)$ to $P_k$ \COMMENT{Move the skyline to $P_k$}
\STATE $P=P-\textit{Skyline}(k)$

\ENDWHILE

\RETURN{ $P_k$}

\end{algorithmic}
\end{algorithm}

After obtaining the ranking, we can assume that the pairs of the first few skylines are more likely to refer to the same physical entity than the rest.

\begin{assumption}
The probability that a pair is labeled positive is inversely proportional to its skyline level. 
\end{assumption}

The assumption considers that for all $ { \langle s_i,s_j \rangle \text{and} \langle s'_{i},s'{j} \rangle}$ in $P$ such that $\langle s_i,s_j \rangle \in \textit{Skyline(k)}$,
$\langle s'_{i},s'_{j} \rangle \in \textit{Skyline(k')}$ and $k<k'$, then $ \langle s_i,s_j \rangle$ is more likely to be a match than $ \langle s'_i,s'_j \rangle$. 

\section{Estimating k}

In this section, we estimate the skyline level $k$ that separates the positive from the negative class. We introduce two different methods for fixing the value of $k$: \emph{threshold-based} (\emph{SkyEx-F} and \emph{SkyEx-FES}) and \emph{unsupervised} (\emph{SkyEx-D}).

\subsection{{SkyEx-F and SkyEx-FES}} 

In contrast to the threshold-based methods used in entity resolution problems \cite{berjawi2014representing, morana2014geobench,quercini2017liaison} where we have to find a threshold for each similarity of the attributes and then a threshold for the similarity function that aggregates the similarity scores, we have simplified our problem to only one parameter: $k$. We need to find the value of $k$ that best separates the classes. As a measure of a "good model", we choose to use the \textit{F-measure}, given that our data tends to be unbalanced \cite{weiss2000learning, cieslak2008learning, zhang2004learning}. In the context of our problem, we define true positives $\textit{TP}$ as pairs that refer to the same physical entity and are correctly labeled as positives; true negatives $\textit{TN}$ as pairs referring to different physical entities and are correctly labeled as negatives; false positive $\textit{FP}$ as pairs that do not refer to the same physical entities but are wrongly labeled as positives; $\textit{FN}$ as  pairs that refer to the same physical entity but are wrongly labeled as negatives. Thus, the precision is $\textit{p} = \frac{\textit{TP}}{\textit{TP}+\textit{FP}}$, the recall is $\textit{r}=\frac{\textit{TP}}{\textit{TP}+\textit{FN}}$ and $\textit{F-measure} \ (F1)=2 \frac{\textit{p} * \textit{r}}{\textit{p} + \textit{r}}$.

\begin{subalgorithms}
\begin{algorithm} [htp]

\caption{SkyEx-F} 
\label{alg:skyex-f}

\begin{algorithmic}[1]
\footnotesize

\INPUT A set of pairs  $P=\{\langle s_i, s_j \rangle \}$ \\ 
\OUTPUT A set of positive pairs $P^+$ and a set of negative pairs $P^-$  ;
\STATE $P_k \gets \emptyset$, $F \gets \emptyset$
\setcounter{ALC@line}{1}

\WHILE{ $\lvert P_k \rvert < \lvert P \rvert$ }

\skipnumber{ Lines 3-5 as Algorithm \ref{alg:pareto} ... }
\setcounter{ALC@line}{5}

\STATE $P^+ \gets P_k$
\STATE $P^- \gets P $
\STATE Calculate $F1(k)$
\STATE Add $\langle k, F1(k) \rangle$ to $F$

\ENDWHILE

\STATE Find $k_f$ such that $F1(k_f) = max(F1(k))$  $\forall k \in \{ 1, \lvert F \rvert \}$

\STATE $P^+ \gets \bigcup_{k=1}^{k_f} Skyline(k)$

\STATE $P^- \gets P - P^+$

\RETURN{ $P^+, P^-$}
\end{algorithmic}

\end{algorithm}

The higher the $k$, the more unlikely it is for a pair in the $k^{th}$ skyline to belong to the positive class (Assumption~1). \emph{SkyEx-F} explores the first skylines and stops at the value of $k=k_f$ that achieves the highest $\textit{F-measure}$. To find $k_f$, we rank the pairs as in Algorithm \ref{alg:pareto}, but we add some extra calculations within the loop (lines 6-9) and find the optimal $k_f$ (line 7) in Algorithm \ref{alg:skyex-f}. \emph{SkyEx-F} calculates the \textit{F-measure} for each skyline $k$ by considering the pairs up to the $k^{th}$ skyline as positive and the rest as negative. We add $F1(k)$ to the set $F$, which keeps track of the evolution of \textit{F-measure} while exploring more skylines. We find $k_f$ as the value of $k$ that achieves the highest \textit{F-measure} in $F$. The pairs from the first to the $k_f$ level of skyline are labeled as positive and the rest as negative. Note that \emph{SkyEx-F} explores all the skylines and then, finds the threshold $k_f$. However, we can optimize Algorithm \ref{alg:skyex-f} by stopping at $k_f$ before going through the full dataset $P$. 
Let us highlight some properties of \textit{p} and \textit{r}.

\begin{property}
The recall is a monotonically non-decreasing function with respect to the number of skylines $k$.
\end{property}

\begin{proof}
The recall after $k$ skylines is $\textit{r (k)}=\frac{\textit{TP(k)}}{\textit{TP(k)}+\textit{FN(k)}}$. While we move to the next, $k+1^{th}$ skyline, we label more pairs as positive, so the probability of finding true positives  $\textit{TP}$ is higher. Thus, $\textit{TP(k+1)}\geq \textit{TP(k)}$. As for the denominator, it is always the same despite the skyline level because the true positives are fixed in $P$ and are independent of our labelling. This means that if we find more true positives ($\textit{TP}$), then we automatically decrease the false negatives ($\textit{FN}$). Hence, $\textit{TP(k+1)}+\textit{FN(k+1)}=\textit{TP(k})+\textit{FN(k)}$. We can then show that $\frac{\textit{TP(k+1)}}{\textit{TP(k+1)}+\textit{FN(k+1)}}\geq\frac{\textit{TP(k)}}{\textit{TP(k)}+\textit{FN(k)}}$ so $\textit{r (k+1)} \geq \textit{r (k)}$.
\end{proof}

\begin{property}
Given Assumption 1, the precision is a monotonically non-increasing function with respect to the number of skylines $k$.
\end{property}

The precision is $\frac{\textit{TP}}{\textit{TP}+\textit{FP}}$. However, $\textit{TP}+\textit{FP}$ is what our algorithm labels as positive, which means all the pairs belonging to skylines up to the $k^{th}$ level. According to Assumption 1, $\textit{FP}$ increase at a higher rate than $\textit{TP}$ while moving to higher $k$ values. A proof of monotonic decreasing precision for systems that rank the results considering their relevance (like our skylines) can be found in \cite{gordon1989recall}.

\begin{theorem}
The F-measure function with respect to the number of skylines $k$ is increasing until a point or interval, and after that, it cannot increase again.
\end{theorem}

\begin{proof}
Let us suppose that while moving deeper into the skylines, we found a peak point $k$ or peak interval $[k_i,k_j]$ with $F1(k)$ as the corresponding \textit{F-measure}. Note that for a peak interval the \textit{F-measure} is constant. Since $F1(k)$ belongs to a peak point/interval, there exists a $F1(k+\epsilon)$ $\epsilon$ skylines after $k$ such that $F1(k+\epsilon)<F1(k)$. Now, let us know suppose that we can find another optimum in $k+\delta$ such that  $F1(k+\delta)>F1(k)$. Since $F1(k+\epsilon)<F1(k)$, consequently $F1(k+\delta)>F1(k)>F1(k+\epsilon)$. $F1=2 \frac{{p} * {r}}{{p} + {r}}$ can be rewritten as $F1= \frac{2}{\frac{1}{p} + \frac{1}{r}}$.
So, we can rewrite: $\frac{2}{\frac{1}{p(k+\delta)} + \frac{1}{r(k+\delta)}}>  \frac{2}{\frac{1}{p(k)} + \frac{1}{r(k)}}>\frac{2}{\frac{1}{p(k+\epsilon)} + \frac{1}{r(k+\epsilon)}}$. Using Property 2, $p(k+\delta) \leq p(k+\epsilon)$, so:$\frac{2}{\frac{1}{p(k+\epsilon)} + \frac{1}{r(k+\epsilon)}} \geq \frac{2}{\frac{1}{p(k+\delta)} + \frac{1}{r(k+\epsilon)}}$, which means that: $\frac{2}{\frac{1}{p(k+\delta)} + \frac{1}{r(k+\delta)}}>\frac{2}{\frac{1}{p(k+\delta)} + \frac{1}{r(k+\epsilon)}}$.
According to Property 1, this inequality cannot hold, because $r(k+\delta)\geq r(k+\epsilon)$. Thus, our assumption of $F1(k+\delta)>F1(k)$ cannot hold and $F1(k)$ remains the highest value of \textit{F-measure}.
\end{proof} \begin{algorithm} [htp]

\caption{{SkyEx-F Early Stop (SkyEx-FES)}} 
\label{alg:skyex-fes}

\begin{algorithmic}[1]
\footnotesize

\INPUT A set of pairs  $P=\{\langle s_i, s_j \rangle \}$ \\ 
\OUTPUT A set of positive pairs $P^+$ and a set of negative pairs $P^-$  ;
\STATE $P_k \gets \emptyset$, $F_{\textit{previous}} \gets 0$
\setcounter{ALC@line}{1}

\WHILE{ $\lvert P_k \rvert < \lvert P \rvert$ }

\skipnumber{ Lines 3-8 as Algorithm \ref{alg:skyex-f}...}
\setcounter{ALC@line}{8}

\IF {$F1(k)<F_{\textit{previous}}$}
\STATE \textbf{break}
\ELSE
\STATE {$F_{\textit{previous}} \gets F1(k)$}
\ENDIF
\ENDWHILE

 \STATE $P^+ \gets \bigcup_{k=1}^{k_f} Skyline(k)$

\STATE $P^- \gets P - P^+$

\RETURN{ $P^+, P^-$}
\end{algorithmic}

\end{algorithm}

\end{subalgorithms}
Theorem 1 ensures that once we find the peak in the \textit{F-measure} function, we can stop finding all the skylines and label the pairs accordingly. Consequently, we can allow Algorithm \ref{alg:skyex-f} to stop early.
The modifications of Algorithm \ref{alg:skyex-f} are reflected in Algorithm \ref{alg:skyex-fes}. We use the same procedure as in Algorithm \ref{alg:skyex-f}, but we do not need to keep track of each of the skylines and their corresponding $\textit{F-measures}$.
Rather, we only keep the previous $\textit{F-measures}$ in $F_{\textit{previous}}$. While moving to the next skyline, we calculate the $\textit{F-measure}$ and the first time  we notice a drop (line~9), we stop the loop (line~10) and return both classes separated by the current $k$ (lines 7-8)). Otherwise, we update $F_{\textit{previous}}$ to the current $\textit{F-measure}$  (line~12) and continue the search for the optimal $k$.

\subsection{SkyEx-D}

The methods described in the previous sections assume that the labels of the pairs are present.
In this section, we assume no information about the labels, and thus, we propose a heuristic for fixing the value of $k$. The heuristic is based on the distance between the positive and the negative class. We refer to the $k$ discovered by \emph{SkyEx-D} as \emph{distance-based k} or $k_d$. Our classes are not characterized by a small intra-class distance. Various patterns can reveal a positive class; for example, a similar name but different category or similar category and similar address, etc. Thus, the positive pairs, positioned in the first skylines, are scattered and do not necessarily form a cluster. However, they can still be separated from the rest, considering the distance to the negative class. Theoretically, the inter-class distance stays small when we are in the first skylines (potential positives), then starts to increase while we move into later skylines and finally falls again when we enter the deeper skylines (potential negatives). \emph{SkyEx-D} notices the increase of the inter-class distance and sets $k_d$ accordingly. In order to have an approximation of the inter-class distance, we use the mean and denote it as $\mu_d$ as in Eq. 3. 
\begin{equation}
    \mu_d = \frac{\sum d(p^k,p^{-k})}{\lvert P_k \rvert}
\end{equation}
\noindent where  $\lvert P_k \rvert $ is the number of pairs from the $1^{st}$ to the $k^{th}$ skyline, $p^k$ is a pair in $P^k$, $p^{-k}$ is a pair in $P-P^k$, and $d(p^k,p^{-k})$ is the distance between $p^k$  and $p^{-k}$. 

In order to fix $k_d$, we monitor the value of  $\mu_d$ while moving deeper into the skylines. We denote by $\mu_d(k)$ the function of $\mu_d$ with regard to $k$. We use the first derivative of  $\mu_d(k)$, denoted as  $\mu'_d(k)$, to  find the points where the  $\mu_d(k)$ function decreases. The intuition behind this approach is that in the beginning, the distance $\mu_d(k)$ starts increasing, which means that the first derivative has a positive slope ($\mu'_d(k)>0$). Later, we enter the "grey area", where there is a mix of potential positives and potential negatives. This is where we need to stop because we might lose precision if we continue further. In order to find the "grey area", we note when the first derivative changes its slope to negative. In order to calculate $\mu'_d(k)$, we estimate the value of $\mu'_d(k)$ in each point $k$ as in Eq. 4:
\begin{equation}
    \mu'_d(k) =  \diffp{}{k} \approx \frac{\mu_d(k+1)-\mu_d(k)}{1}
\end{equation}
In order to not be sensitive to small fluctuations in $\mu'_d(k)$, we smoothen slightly  $\mu'_d(k)$ with Gaussian function ($\frac{1}{{\sigma \sqrt {2\pi } }}e^{{{ - \left( {x - \mu } \right)^2 } \mathord{\left/ {\vphantom {{ - \left( {x - \mu } \right)^2 } {2\sigma ^2 }}} \right. \kern-\nulldelimiterspace} {2\sigma ^2 }}}$) using a small window. Then we monitor when $\mu'_d(k)$ decreases for the first time and we set $k_d$ accordingly. We modify Algorithm \ref{alg:pareto} to accommodate this approach. We calculate  $\mu'_d(k)$ for each point of $k$ in line~7.  Then, we have to find the first negative value of  the smoothened $\mu'_d(k)$ (line~9) and fix $k_d$ accordingly (line~10). Finally, we return the classes defined by $k_d$ in lines 16-17.
\begin{algorithm} [htp]

\caption{SkyEx-D} 
\label{alg:skyex-d}

\begin{algorithmic}[1]
\footnotesize

\INPUT A set of pairs  $P=\{\langle s_i, s_j \rangle \}$ \\ 
\OUTPUT A set of positive pairs $P^+$ and a set of negative pairs $P^-$  ;

\STATE $P_k \gets \emptyset$

\skipnumber{ Lines 2-6 as Algorithm \ref{alg:pareto}...}

\setcounter{ALC@line}{6}

\STATE Calculate  $\mu'_d(k)$ in each $k$

\WHILE {$k<k_{\textit{last}}$  }
\IF { $\textit{smooth}(\mu'_d(k))<0$}
\STATE $k_d \gets k$
\STATE \textbf{break}
\ELSE 
\STATE $k \gets k+1$
\ENDIF
\ENDWHILE

\STATE $P^+ \gets \bigcup_{k=1}^{k_d} Skyline(k)$
\STATE $P^- \gets P_k - P^+$

\RETURN{ $P^+, P^-$}
\end{algorithmic}

\end{algorithm}

\emph{\textbf{Summary. }}
Algorithm \ref{alg:skyex-d} estimates the skyline level $k$ that best separates the positive class from the negative class. \emph{Similarly to clustering techniques that use heuristics to estimate their parameters, SkyEx-D uses the distance of the positive class from the rest as an indicator of class separability}. However, in contrast to clustering metrics, which focus on the robustness of clusters, this is not a requirement for the \emph{SkyEx-*} family of algorithms. \emph{The positive pairs do not show similar patterns, but rather similar utilities, which can be better captured by skylines} (see Sect. 9.9). Experimentally, we show that our inter-class distance approach estimates $k_d$ very close to $k_f$ without loosing in \textit{F-measure}. In contrast to techniques that use a scoring function, the \emph{SkyEx-*} family of algorithms abstracts the concept of utility. Thus, no weights or similarity function is needed. Even though the positive class can be characterized by various patterns of attribute similarities,  the \emph{SkyEx-*} family of algorithms can still group together the positive class based on the high utility, while a clustering technique would instead focus in grouping each pattern separately, without putting the positive-class pairs together into one cluster. Moreover, the flexibility of the \emph{SkyEx-*} family of algorithms makes it applicable to all problems where the expert knowledge on the contribution of the attributes is missing. Finally, the \emph{SkyEx-*} family of algorithms does not learn any behavior, so there is no risk of overfitting.

\section{Complexity Analysis of QuadSky}

In this section, we discuss the time complexity of our algorithms and of our \emph{QuadSky} solution.

\textbf{\emph{QuadFlex}} deals with points (not regions); thus, it behaves similarly to a point quadtree. \emph{QuadFlex} splits the same way as a quadtree, but in contrast to the quadtree, the points can be assigned to more than one child.
We construct the \emph{QuadFlex} structure only for forming the blocks. Hence, the construction complexity is of interest to us. Let us denote by $\lvert S \rvert$  the number of points in $S$, $c$ the smallest distance between any two points, and $D_1$ and $D_2$ the dimensions of the initial area containing all the points. 
Let us first estimate the depth of \emph{QuadFlex}. The distance $c$ of any two points $p_1$ and $p_2$ in \emph{QuadFlex} is always less than the diagonal of the node they belong in. Given that \emph{QuadFlex} allows neighboring points to be included in more than one child, this calculation needs to be modified. The physical diagonal of the initial (level 0) node is $\sqrt{D_1^2 + D_2^2}$. The diagonal of level $i$ is $\frac{\sqrt{D_1^2 + D_2^2}}{4^i}$. To modify the calculation, we estimate the logical diagonals if \emph{QuadFlex} would physically expand to accommodate neighboring points, so: $c \leq  \frac{\sqrt{\frac{3D_1^2}{2} + \frac{3D_2^2}{2}}}{4^i} $
Now, isolating $i$ out of this equation results in 
$i \leq  \log_4 \frac{\sqrt{\frac{3}{2} (D_1^2+D_2^2)}}{c} = \log_4 \sqrt{\frac{3}{2}}+\log_4\frac{\sqrt{D_1^2+D_2^2}}{c}$.
   $\log_4 \sqrt{\frac{3}{2}} \approx 0.14$  so we can discard it (less than one level):
$i \leq  \log_4\frac{\sqrt{D_1^2+D_2^2}}{c}$. For estimating the maximal depth, we need to add one more level (root)
    so the depth is estimated as $\log_4\frac{\sqrt{D_1^2+D_2^2}}{c} + 1$.
    Finally, for constructing \emph{QuadFlex}, the complexity is $O(\lvert S \rvert (\log_4\frac{\sqrt{D_1^2+D_2^2}}{c} + 1))$.

    \textbf{\emph{SkyRank}}  requires calculating the Pareto frontiers, which is time-consuming. In the typical case, comparing the pairs $P$ resulting from \emph{QuadFlex} in terms of all $d$ dimensions has a  $O({2^{\lvert P \rvert}} ^d)$ time complexity, 
    which is not scalable. \emph{SkyRank} uses the method proposed in \cite{endres2015scalagon}, which first scales down the d-dimensional domain and then pre-filters the data using a lattice. This yields a time complexity of $O({\lvert P \rvert} ^2)$ for the first skyline.
    For the total number of $K$ skylines, the complexity is $O(K{\lvert P \rvert} ^2)$.
    
    \textbf{\emph{SkyEx-F}} calculates the metrics while adding the next skyline to the positive class; thus, these calculations do not add any complexity. Finally, we perform a linear search on $F$ to find the skyline with the highest F-measure. The size of $F$ is equal to $K$, so the complexity is $O(K{\lvert P \rvert} ^2 + K)$.

     \textbf{\emph{SkyEx-FES}} stops earlier than \emph{SkyEx-F}, avoiding a big part of the time-consuming Pareto calculations. Given that the best pairs usually are focused on the first skylines, the cut-off $k \ll K$. Moreover, according to Theorem 1, we do not need to store $F$, so we avoid the linear search for the best F-measure. The complexity is $O(k{\lvert P \rvert} ^2)$.

     \textbf{\emph{SkyEx-D}} uses all $K$ Pareto calculations and then, in order to estimate the cut-off $k_d$, it computes the distance between the positive class and the rest. \emph{SkyEx-D} creates a matrix where the rows are the positive class $P^+$ and the columns are the negative class data points $\lvert P \rvert - P^+$, so the complexity is $P^+ * (\lvert P \rvert - P^+)$.
     $P^+ * (\lvert P \rvert - P^+) = P^+ * \lvert P \rvert - {(P^+)}^2 $ is the equation of a vertical parabola that opens downwards $-ax^2 + bx + c$, with the maximal value at the vertex $(-\frac{b}{2a})$. In our case, the maximum of $P^+ * (\lvert P \rvert - P^+)$ is at 
     $\frac{\lvert P \rvert} {2}$, resulting in a maximal complexity of $\frac{\lvert P \rvert ^2} {4}$. For each skyline $k$ in $K$, the maximal complexity is   $\frac{\lvert P \rvert ^2} {4}$, thus, $K \frac{\lvert P \rvert ^2} {4}$ for all.  
     Note here that $K\ll\lvert P\rvert $ so it is far from a cubic complexity. SkyEx-D computes the mean distance $\mu_d$ for each $k$, which can already be done within the  $\frac{\lvert P \rvert ^2} {4}$ complexity. Then, we compute the derivative $\mu'_d$  on the means, which has a linear complexity in  $K$. Finally, we need another partial scan until $k_d$ ($k_d \ll K$) where the derivative  $\mu'_d$ becomes negative for the first time. Hence, the total complexity is $ O (K{\lvert P \rvert}^2 + K \frac{\lvert P \rvert ^2} {4} + K + k_d) = O(\frac{5K}{4}\lvert P \rvert + K + k_d$).
    
    \textbf{\emph{Summary}}. \emph{QuadFlex} has $O(n \log n)$ complexity, the pairwise comparisons have a linear $O(n)$ complexity, while the \emph{SkyEx-*} family of algorithms have a quadratic complexity $O(n^2)$. However, there is a theoretical risk of a cubic complexity in \emph{SkyEx-F} and \emph{SkyEx-D} if the number of skylines $K = \lvert P \rvert$. This means that each skyline in $K$ contains only one pair of entities, which theoretically can happen but almost never happens in practice. Thus, the algorithms have quadratic complexity in the average case. \emph{SkyEx-D} has the highest complexity, followed by \emph{SkyEx-F} and \emph{SkyEx-FES}. Overall, \emph{QuadSky} has a quadratic complexity.

\section{Experiments}

\subsection{Dataset Description}

The spatial entities that will be used in these experiments originate from four sources, namely Google Places (GP), Foursquare (FSQ), Yelp, and Krak. Krak (www.krak.dk) is a location-based source that offers information about companies, enterprises, etc. in Denmark and is also part of Eniro Danmark A / S., which publishes The Yellow Pages. The data is obtained by using the available APIs and the algorithm detailed in \cite{isaj2019seed}. 
The dataset consists of 75,541 spatial entities where 51.50\%  comes from GP, 46.22\% from Krak, 0.03\% from FSQ, and 2.23\% from Yelp (see Annex A for the spread of these spatial entities on the map). The dataset is 69 MB. For a 100 m blocking, there are 35,521 spatial entities that have at least one positive match in the dataset, resulting in 27,102 pairs that need to be discovered. 7,795 of these pairs are within the same source, which shows that none of these sources are free of duplicates. 3,546 of the same-source links come from GP, 3,789 from Krak, and 460 from Yelp. As for the different-source links, all the sources overlap with each other, but the highest overlap of 17,405 pairs (90\% of different-source links) comes from Krak and GP.

\subsection{QuadFlex Performance}
\begin{figure}[t]
  \begin{subfigure}[b]{0.46\linewidth}
    \includegraphics[width=\linewidth, trim={0 6mm 0 2mm},clip]{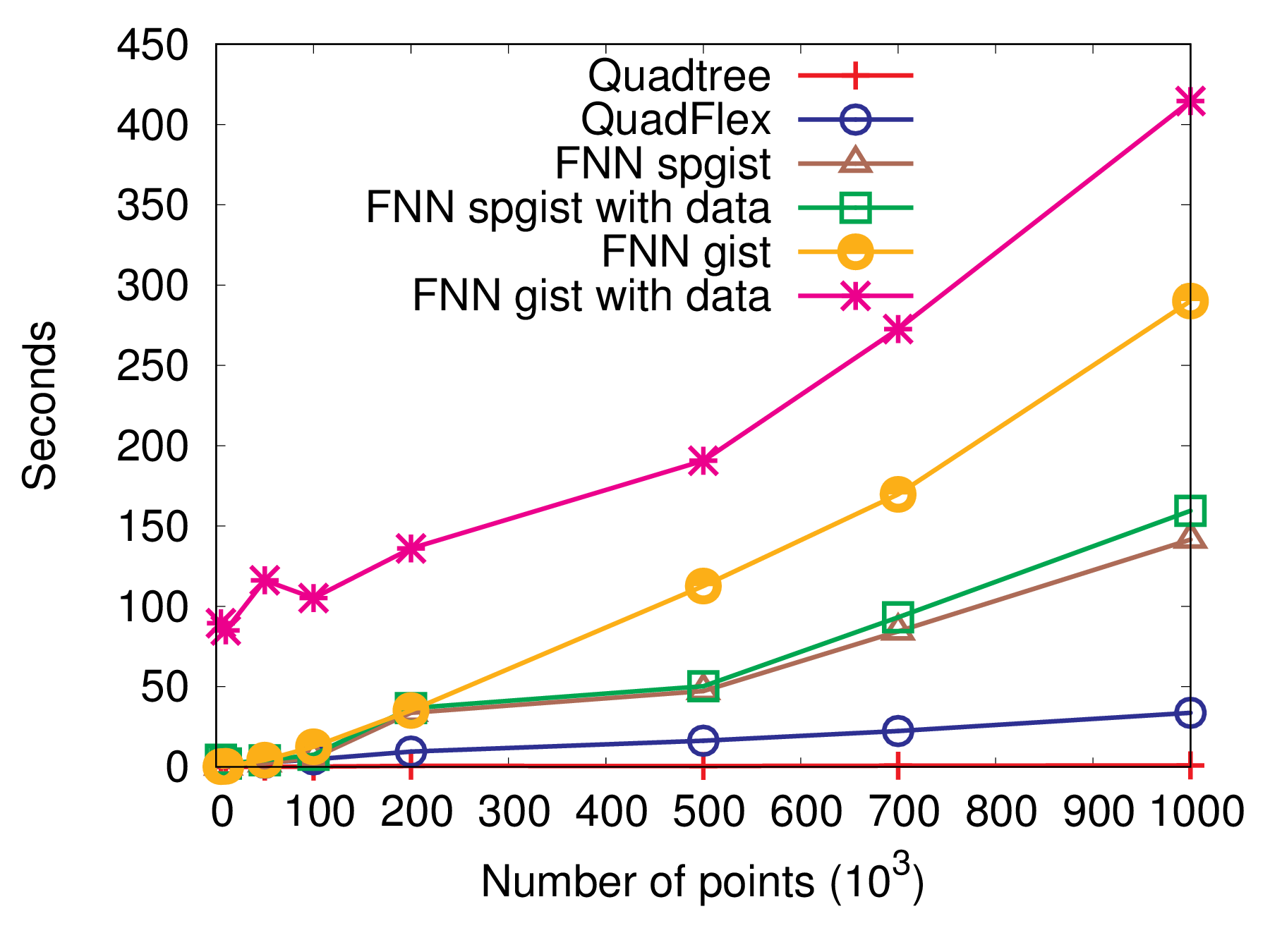}
   \caption{Execution time}
    \label{fig:timeQuad}   
  \end{subfigure}\begin{subfigure} [b]{0.46\linewidth}
     \includegraphics[width=\linewidth, trim={0 6mm 0 2mm},clip]{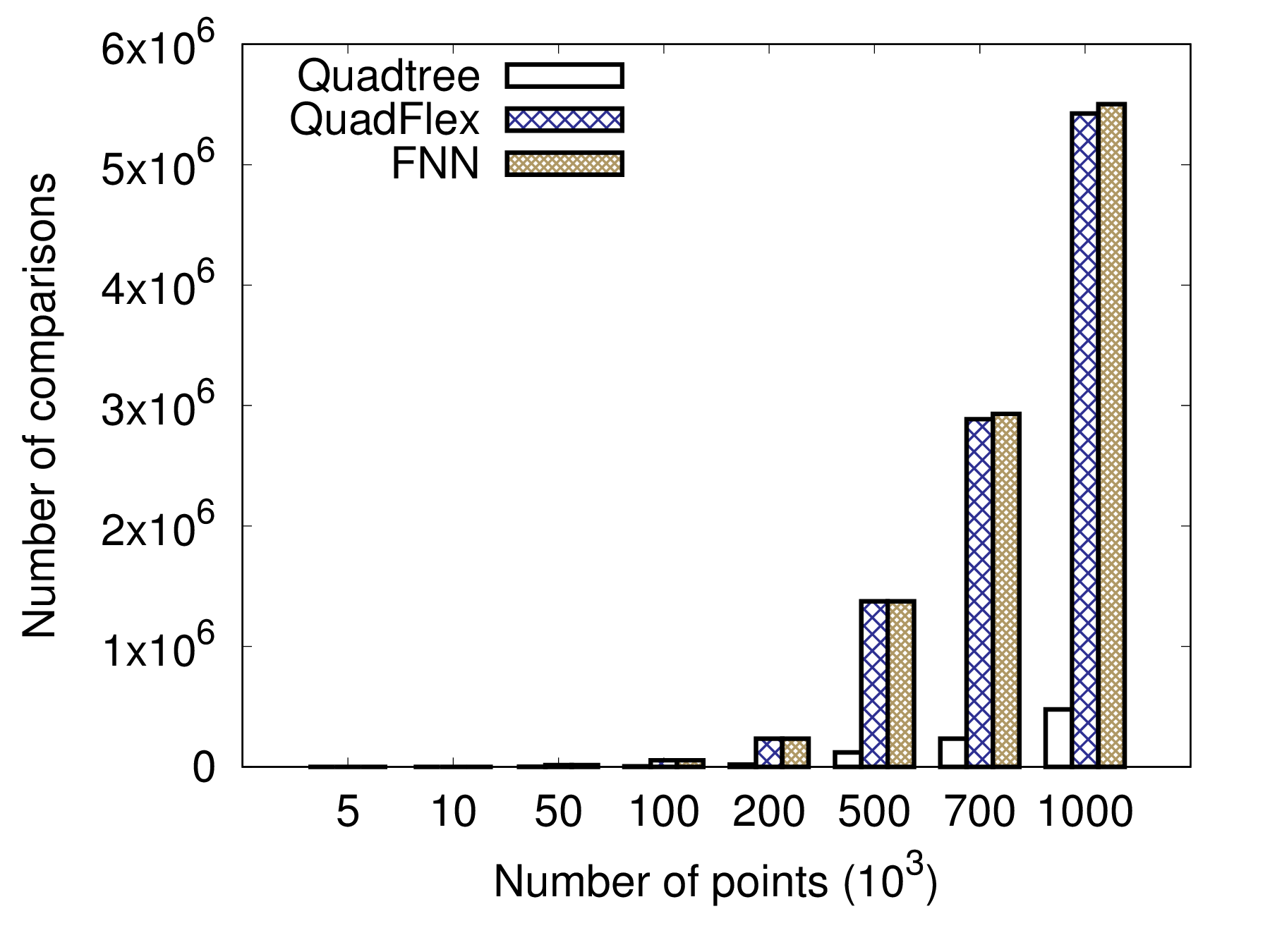}
   \caption{Number of comparisons}
    \label{fig:completeQuad}
 \end{subfigure} 
 \caption{Comparing quadtree, QuadFlex and FNN}
 \label{fig:compQuad}
\end{figure} 

\begin{figure*}[htp]
 \begin{minipage}{.5\textwidth}
  \centering
  \begin{subfigure}[b]{0.48\linewidth}
    \includegraphics[width=\linewidth, trim={0 6mm 0 4mm},clip]{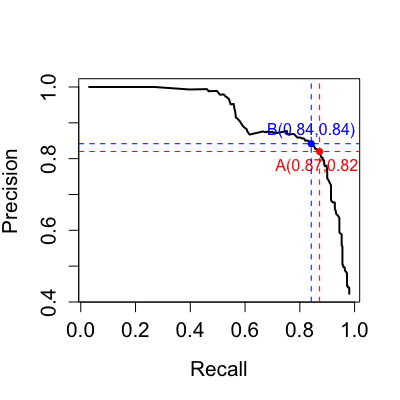}
   \caption{precision and recall}
    \label{fig:samplePR}   
  \end{subfigure}\begin{subfigure} [b]{0.48\linewidth}
     \includegraphics[width=\linewidth, trim={0 6mm 0 4mm},clip]{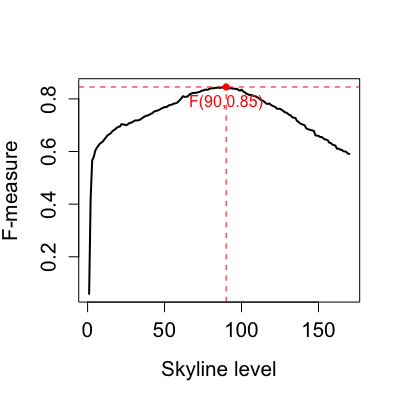}
   \caption{F-measure}
    \label{fig:sampleF}
 \end{subfigure} 
 
 \caption{SkyEx-F performance on $D_{\textit{sample}}$
 }
 \label{fig:sampleEval}
 \end{minipage}\begin{minipage}{.5\textwidth}
  \centering
 \begin{subfigure}[b]{0.48\linewidth}
    \includegraphics[width=\linewidth, trim={0 6mm 0 4mm},clip]{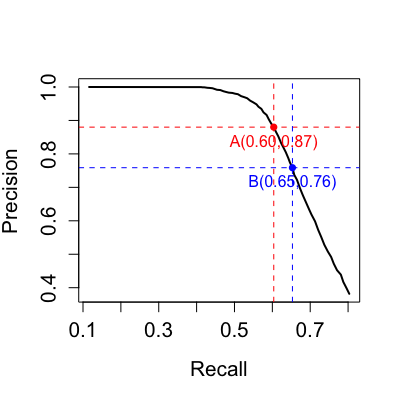}
   \caption{precision and recall}
    \label{fig:fullPR}   
  \end{subfigure}\begin{subfigure} [b]{0.48\linewidth}
     \includegraphics[width=\linewidth, trim={0 6mm 0 4mm},clip]{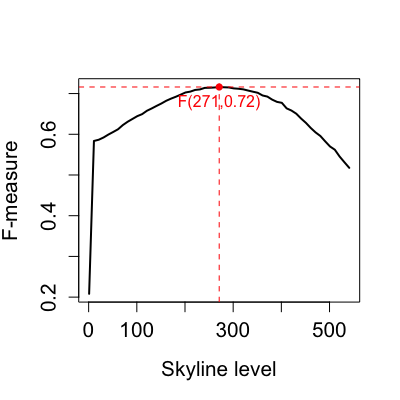}
   \caption{F-measure}
    \label{fig:fullF}
 \end{subfigure} 
 \caption{SkyEx-F performance on $D_{\textit{full}}$
 }
 \label{fig:fullEval}
  \end{minipage}
\end{figure*}

In this section, we compare the performance of \emph{QuadFlex} to the quadtree, Fixed Radius Nearest Neighbors algorithm \cite{bentley1977complexity} (FNN), and having no index at all (No-Index). FNN finds the neighbors that fall within a fixed radius from each point. \emph{QuadFlex} and the quadtree algorithm are implemented in Java, while FNN is run on a Postgres database (https://www.postgresql.org) using spatial indexes: GiST  
(optimized C implementation of B-trees and R-trees) and SP-GiST 
(optimized C implementation of quadtrees and k-d trees). Our dataset contains 75,541 entities in the North Denmark region (around 16 towns, 7,933 $km^2$), so the average density is not high, even though there are areas with high density. A high data density means more pairs to compare. To test our \emph{QuadFlex} on different data densities, we simulate up to 1,000,000 random points from Aalborg (139 $km^2$). Fig.~\ref{fig:compQuad} shows the comparison of quadtree, \emph{QuadFlex} and FNN in terms of execution time (Fig.~\ref{fig:timeQuad}) and number of comparisons (Fig.~\ref{fig:completeQuad}). The FNN versions with data are computed on the database, and then the pairs are loaded back to the java implementation.  The quadtree has the lowest execution time, followed by \emph{QuadFlex}. FNN SP-GiST is comparable and sometimes even better than \emph{QuadFlex} for small datasets. However, when the size of the dataset increases, \emph{QuadFlex} maintains an execution time that is eight times less than FNN GiST and 3 times less than FNN SP-GiST. FNN with SP-GiST index outperforms FNN GiST for all dataset sizes. No-Index was very inefficient,
up to 848 times slower than FNN Gist with data, and up to 368,095 times slower than \emph{QuadFlex}. Given that No-Index would have dwarfed the other curves, it is not part of Fig.~\ref{fig:timeQuad}, but instead, refer to Fig. 2 in  Annex B.  As for the number of comparisons, \emph{QuadFlex} enumerates 12 times more comparisons than quadtree. Moreover, \emph{QuadFlex} contains almost all (99.99\%) comparisons of FNN, compared to the quadtree that contains only 10\% of FNN. Furthermore, given that the scalability of \emph{QuadFlex} is better than FNN, and \emph{QuadFlex} is independent of the database implementations, the loss of around 0.01\% of comparisons is insignificant.


\subsection{SkyEx-F Results}

 \begin{figure}[t]
 \centering
  \begin{subfigure}[b]{0.45\linewidth}
    \includegraphics[width=\linewidth]{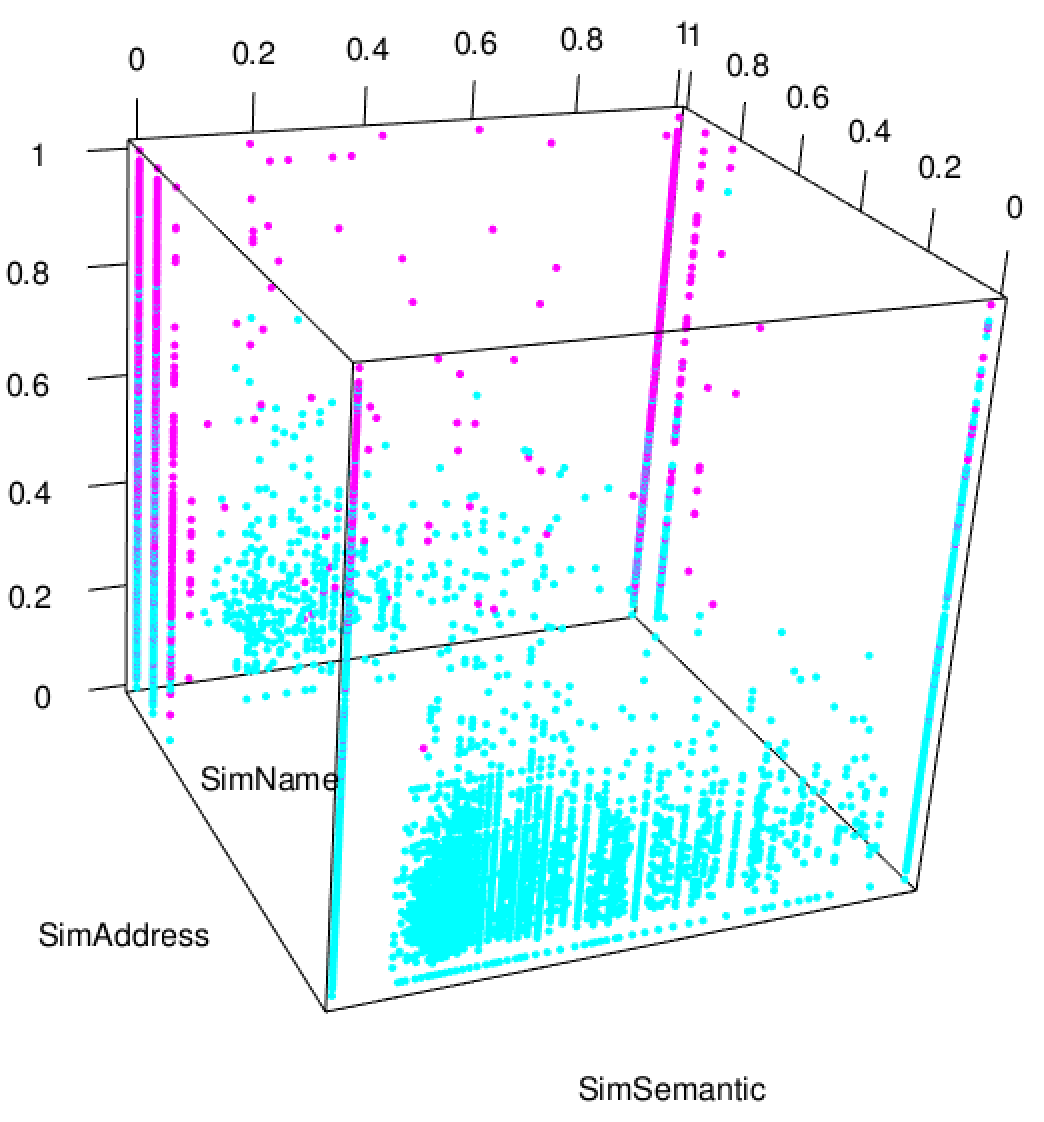}
   \caption{Actual classes}
    \label{fig:classes}   
  \end{subfigure}\begin{subfigure} [b]{0.45\linewidth}
     \includegraphics[width=\linewidth]{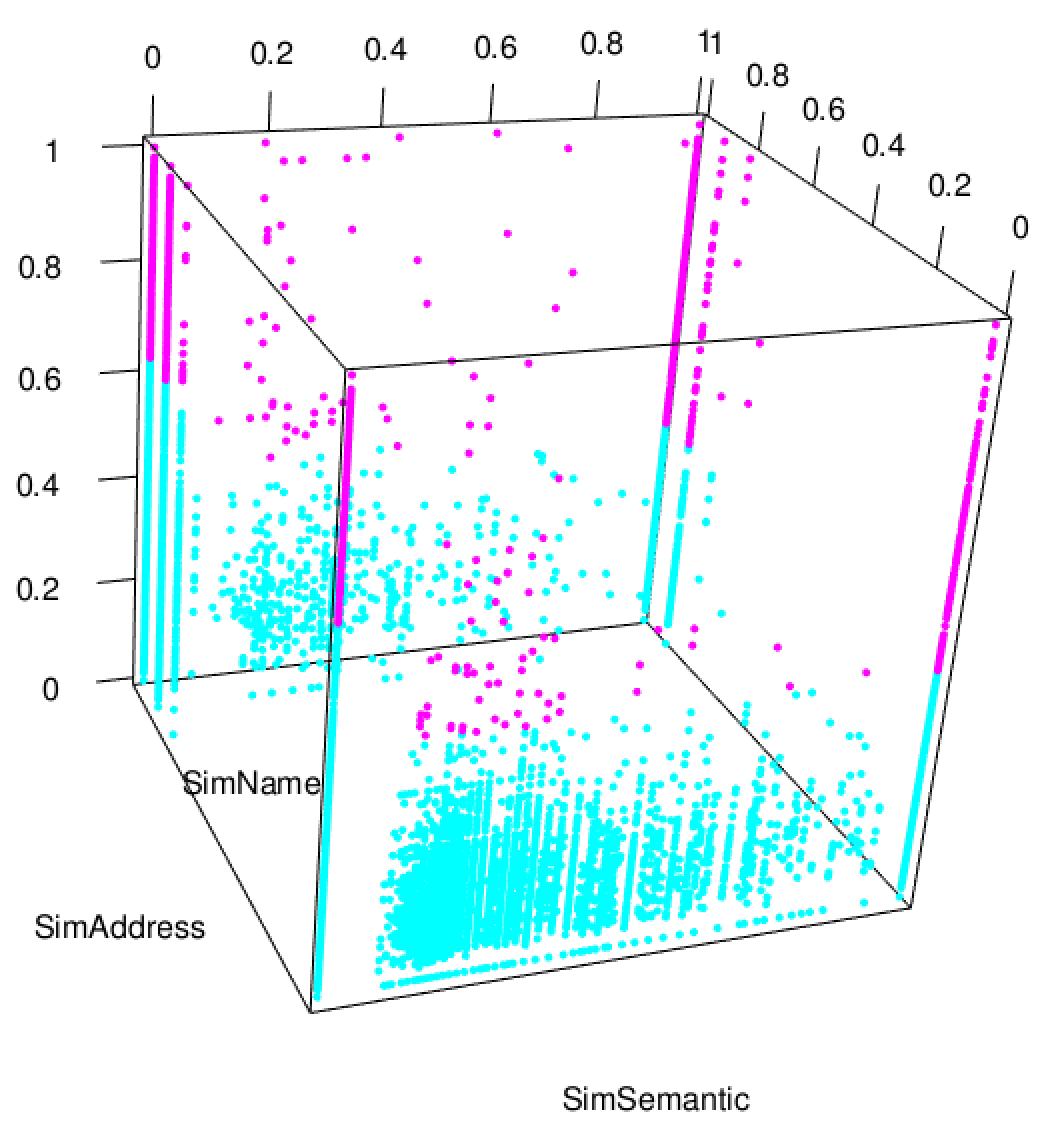}
   \caption{SkyEx classes}
    \label{fig:skyexClasses}
 \end{subfigure} 
 \caption{Positive (in pink) versus negative (in sky blue) classes for actual (a) and SkyEx-F (b) results}
 \label{fig:3D}
 
\end{figure}

We ran \emph{QuadFlex} with 100 m and no density restriction, and we obtained 777,452 pairs (1426 MB). Having the same website or phone is a strong indicator of a match, so we use these attributes to infer the label. We refer to this labeling as \emph{automatic labeling}. However, cases with different phone number or website but still the same entity, or same phone number but different entity might occur. 
Hence, we manually checked the labels of a sample of 1,500 pairs of entities (1552 kB). We will refer to the sample of manually checked pairs as $D_{\textit{sample}}$ and to the full dataset as $D_{\textit{full}}$. Checking the labels manually on the full dataset of 777,452 pairs is unfeasible. Hence, we checked around 10,000 of the pairs, and for the rest, we rely on automatic labeling.

The results of \emph{SkyEx-F} on $D_{\textit{sample}}$ and $D_{\textit{full}}$ are presented in Figs.~\ref{fig:sampleEval} and \ref{fig:fullEval}. The curves in Figs.~\ref{fig:samplePR} and \ref{fig:fullPR} shows the evolution of $\textit{p}$ ($y$-axis) and $\textit{r}$ ($x$-axis) while we move from one skyline to the next. \emph{The more we explore, the more likely it is to retrieve more true positives and thus, improve the $\textit{r}$}. However, \emph{the more we explore and label pairs as positives, the more likely it is to increase the number of false positives, so the $\textit{p}$ degrades}. The algorithm explores several trade-offs; for example the points $A$ and $B$ are among the best. The point $A$ with 0.87 $\textit{p}$ and 0.82 $\textit{r}$ in Fig.~\ref{fig:samplePR} is the same best point in terms of \textit{F-measure}  as well, so that is where \emph{SkyEx-F} will fix $k_f$.  Fig.~\ref{fig:sampleF} shows the levels of the skyline, and the value of \textit{F-measure}  achieved. The highest value is 0.85 that corresponds to $k=90$. 
The evaluation on the full dataset yields lower values (\textit{F-measure} of 0.72) compared to the sample (\textit{F-measure} of 0.85), which might be a simple consequence of automatic labeling. 
Point $A$ has 0.6 $\textit{r}$ and 0.87 $\textit{p}$, while $B$ offers a higher $\textit{r}$ of 0.65 but a lower $\textit{p}$ of 0.76 (Fig.~\ref{fig:fullF}). To have an idea of the real classes in  $D_{\textit{full}}$ and the skylines, we plotted their distribution in Fig.~\ref{fig:3D} (the actual positive classes in pink and the negative ones in sky blue). It is noticeable that the positive class pairs are allocated in the highest values of the dimensions. 
Despite the differences between both plots, \emph{SkyEx-F} shows promising results in separating the positive class from the negative one with 0.6 $\textit{r}$ and 0.87 $\textit{p}$.

\subsection{Experimenting with Different QuadFlex Parameters}

So far, we used \emph{QuadFlex} blocking technique with 100 meters and no density restriction. In this section, we will evaluate our approach \emph{QuadSky} for different blocking parameters. 

\textbf{Changing m, no density limit } In this experiment, we test different values of $m$ used in \emph{QuadFlex} for creating spatial blocks. We test $m$ values of 1, 20, 40, 60, 80, and 100 meters. The size of the dataset for each of them is presented in Table~\ref{tab:mNod}. The spatially closeby points are likely to be a match. Hence, the percentage of the true positives is generally higher for smaller values of $m$. An interesting case is $m=1$, where the percentage of the true positives (\textit{TP}) is lower than $m=20$. One would expect that points that are $1$ meter apart would unquestionably be a match. However, this is not always the case. Shopping malls, buildings that host several companies, etc. are characterized by the same coordinates but not necessarily the same spatial entities. The results for different values of $m$ are presented in Table \ref{tab:mNod} (see the precision-recall graphs for all cut-offs in  Annex C). 
For all cases, the $\textit{r}$ is higher than 0.6. The $\textit{p}$ is higher than 0.8 for all values of $m$, except $m=1$, where the $\textit{p}$ is 0.67.
For $m=1$, the positive and negative class are mixed, thus \emph{SkyEx} loses a bit in $\textit{p}$. This is also an argument against the works that merge arbitrarily points that are 5 m apart. \emph{Spatial proximity is not a definitive indicator of a match}.

\textbf{Changing d, m $\boldsymbol{\leq}$ 100. } We experiment with different values of density $d$ and its effect on the results. 
The size of the dataset, the percentage of the true positives, and the results in terms of precision, recall, and \emph{F-measure} are in
Table~\ref{tab:dNod} (see the precision-recall graphs for all cut-offs in Fig. 9 in \cite{Isaj2019multi}). When the density is smaller, we force \emph{QuadFlex} to split further and create smaller blocks. Thus, the number of pairs reduces. Note that, on the contrary, the percentage of the true positives (\textit{TP}) increases. Indeed, further splits allow us to create better blocks containing a higher percentage of \textit{TP}. However, when the density limit increases above $\frac{30s}{1000m^2}$, fewer and fewer blocks are split further, so the dataset size and the percentage of the \textit{TP} do not vary significantly.
In all the cases, the $\textit{r}$ stays above 0.61 and the $\textit{p}$ above 0.87. A slightly better $\textit{p}$ (0.88) and $\textit{r}$ (0.63) is achieved in the case of a density of $\frac{10s}{1000m^2}$ (the lowest parameter). \emph{SkyEx-F adapts very well in finding the correct classes even when the size of blocks changes and even when the percentage of the true positives over the true negatives varies}.

\begin{table}[t]
\centering
\scriptsize
\caption{\emph{SkyEx-F} results for different m }
\begin{tabular}{@{}lllllll@{}}
\toprule
\textbf{Meters}          & \textbf{$1$} & \textbf{$20$} & \textbf{$40$} & \textbf{$60$} & \textbf{$80$} & \textbf{$100$} \\ \midrule
\textbf{Total}    & 41053      & 118437      & 226331      & 372553      & 557421      & 777452       \\
\textbf{\% of TP} & 17.11\%    & 19.88\%     & 11.28\%     & 7.06\%      & 4.82\%      & 3.49\%       \\
\textbf{Prec.} & 0.67 & 0.80 &0.85 & 0.85 & 0.88 & 0.87 \\
\textbf{Rec.} & 0.60 & 0.69 & 0.65 & 0.64 & 0.61 & 0.61 \\
\textbf{F1} & 0.64 & 0.75 & 0.74 & 0.73 & 0.72 & 0.72 \\
 \bottomrule
\end{tabular}
\label{tab:mNod}

\end{table}

\begin{table}[t]
\centering
\scriptsize
\caption{\emph{SkyEx-F} results for different d }
\begin{tabular}{@{}lllllll@{}}
\toprule
\textbf{Den.}  & \textbf{$\frac{10 s}{1000 m^2}$} & \textbf{$\frac{20 s}{1000 m^2}$}&\textbf{$\frac{30 s}{1000 m^2}$} & \textbf{$\frac{40 s}{1000 m^2}$} & \textbf{$\frac{50 s}{1000 m^2}$}& \textbf{$\frac{60 s}{1000 m^2}$} \\ \midrule
\textbf{Total}       & 290653      & 590583      & 711423      & 754195      & 770987      & 776664      \\
\textbf{\% of TP}    & 8.61\%      & 4.57\%      & 3.81\%      & 3.59\%      & 3.51\%      & 3.49\%      \\
\textbf{Prec.} & 0.88 & 0.88 & 0.87 & 0.87 & 0.87 & 0.87 \\
\textbf{Rec.} & 0.63 & 0.61 & 0.61 & 0.61 & 0.61 & 0.61 \\
\textbf{F1} & 0.74 & 0.74 & 0.72 & 0.72 & 0.72 & 0.72 \\
\bottomrule
\end{tabular}

\label{tab:dNod}

\end{table}

\subsection{SkyEx-FES Optimization}

\begin{figure*}[htbp]
  \begin{subfigure}[b]{0.25\linewidth}
    \includegraphics[width=\linewidth,trim={0 4mm 0 10mm},clip]{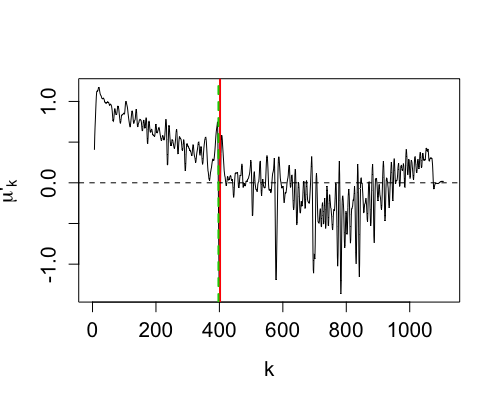}
   \caption{m=30}
    \label{fig:deriv30m}   
  \end{subfigure}\begin{subfigure} [b]{0.25\linewidth}
     \includegraphics[width=\linewidth,trim={0 4mm 0 10mm},clip]{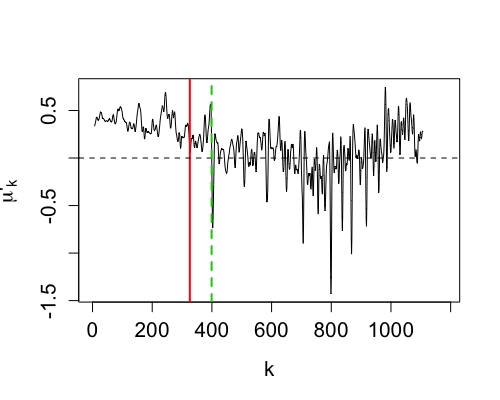}
   \caption{m=50}
    \label{fig:deriv50m}
 \end{subfigure}\begin{subfigure}[b]{0.25\linewidth}
    \includegraphics[width=\linewidth,trim={0 4mm 0 10mm},clip]{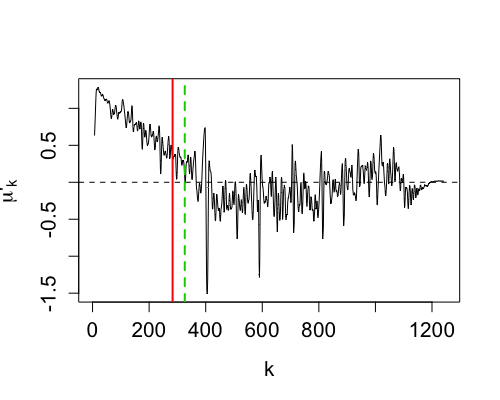}
   \caption{m=80}
    \label{fig:deriv80m}   
  \end{subfigure}
  \begin{subfigure} [b]{0.25\linewidth}
     \includegraphics[width=\linewidth,trim={0 4mm 0 10mm},clip]{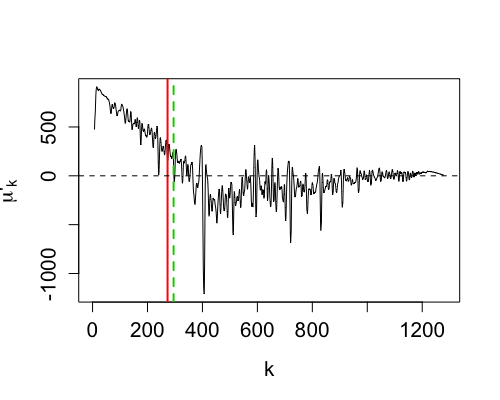}
   \caption{m=100}
    \label{fig:deriv100m}
 \end{subfigure} 
 \caption{Setting $k_d$ using $\mu'_d$}
 \label{fig:deriv}
 \vspace{-4mm}
\end{figure*}

\begin{figure*}[htbp]
  \begin{subfigure}[b]{0.25\linewidth}
    \includegraphics[width=\linewidth,trim={0 4mm 0 10mm},clip]{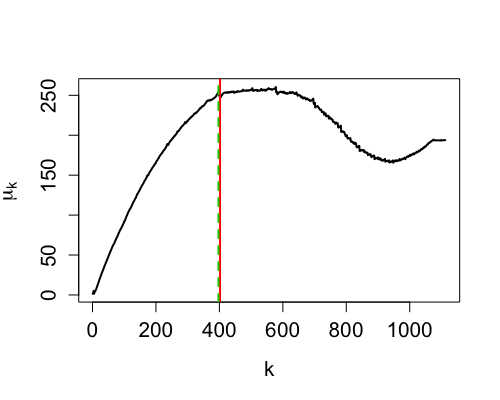}
   \caption{m=30}
    \label{fig:estimk30m}   
  \end{subfigure}\begin{subfigure} [b]{0.25\linewidth}
     \includegraphics[width=\linewidth,trim={0 4mm 0 10mm},clip]{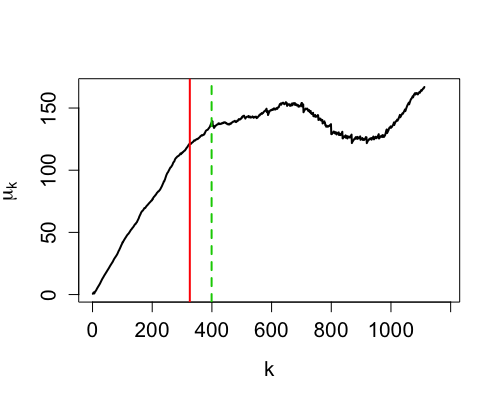}
   \caption{m=50}
    \label{fig:estimk50m}
 \end{subfigure}\begin{subfigure}[b]{0.25\linewidth}
    \includegraphics[width=\linewidth,trim={0 4mm 0 10mm},clip]{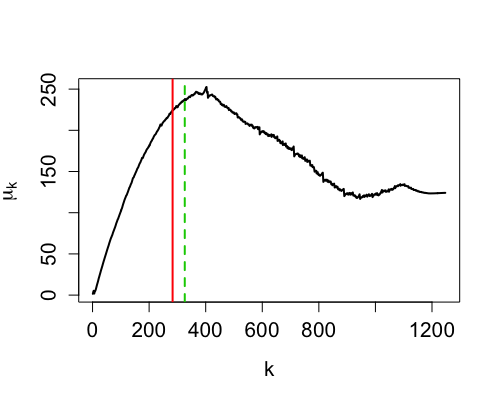}
   \caption{m=80}
    \label{fig:estimk80m}   
  \end{subfigure}
  \begin{subfigure} [b]{0.25\linewidth}
     \includegraphics[width=\linewidth,trim={0 4mm 0 10mm},clip]{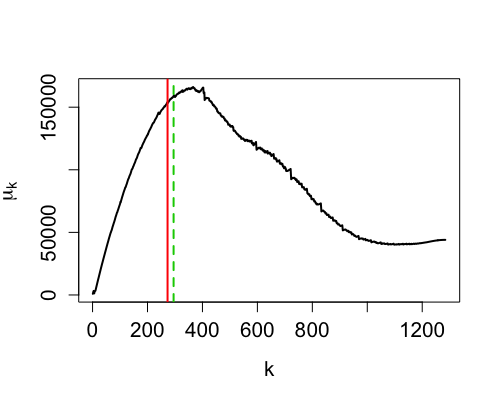}
   \caption{m=100}
    \label{fig:estimk100m}
 \end{subfigure} 
 \caption{$\mu_d(k)$ function with respect to $k$}
 \label{fig:estimk}
 \vspace{-4mm}
\end{figure*}

\begin{figure*}[htbp]
  \begin{subfigure}[b]{0.25\linewidth}
    \includegraphics[width=\linewidth,trim={0 4mm 0 10mm},clip]{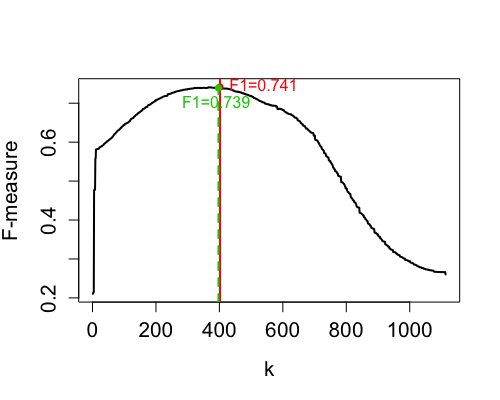}
   \caption{m=30}
    \label{fig:f30}   
  \end{subfigure}\begin{subfigure} [b]{0.25\linewidth}
     \includegraphics[width=\linewidth,trim={0 4mm 0 10mm},clip]{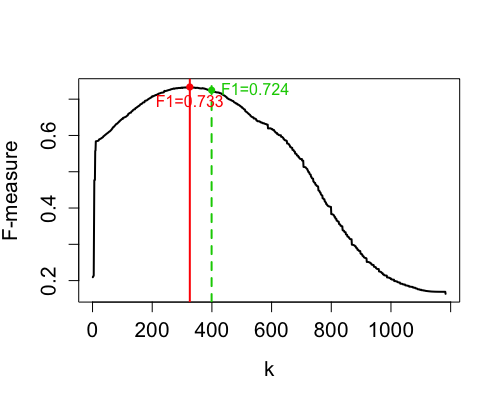}
   \caption{m=50}
    \label{fig:f50}
 \end{subfigure}\begin{subfigure}[b]{0.25\linewidth}
    \includegraphics[width=\linewidth,trim={0 4mm 0 10mm},clip]{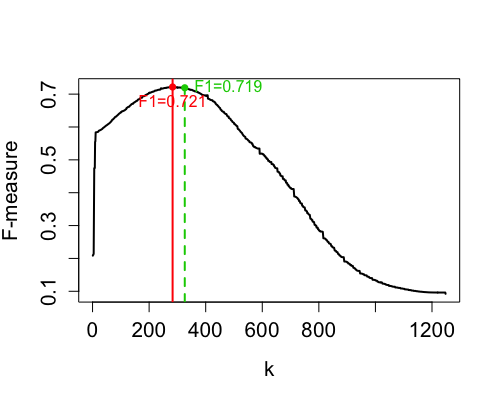}
   \caption{m=80}
    \label{fig:f80}   
  \end{subfigure}
  \begin{subfigure} [b]{0.25\linewidth}
     \includegraphics[width=\linewidth,trim={0 4mm 0 10mm},clip]{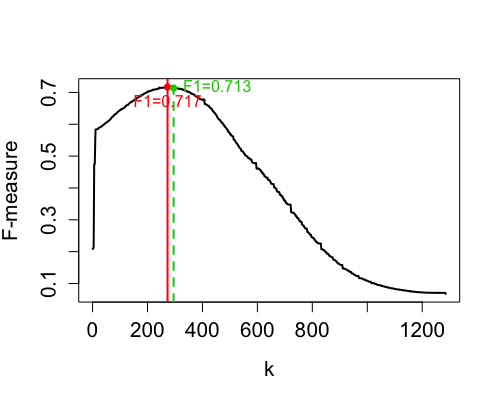}
   \caption{m=100}
    \label{fig:f100}
 \end{subfigure} 
 \caption{\textit{F-measure} values for different $k$}
 \label{fig:f}
 
\end{figure*}

Given the theoretical guarantee in Theorem 1, we can stop \emph{SkyEx-F} earlier  as described in Algorithm \ref{alg:skyex-fes}. We ran \emph{SkyEx-FES} for spatial entities that are 30, 50, 80, and 100 m apart. For all the cases, \emph{SkyEx-FES found the same $k_f$ values as SkyEx-F exploring  only 27\% of the skylines on average.}
The comparison regarding the number of iterations is shown in Table \ref{tab:skyex-fes}. For spatial entities that are 30 m, 50 m, 80 m, and 100 m apart, \emph{SkyEx-FES} finds the optimal $k_f$ exploring 36\%, 27\%, 23\%, and 22\%  of the skylines, respectively. Moreover, our theoretical guarantee that the \emph{F-measure} function has only one optimum can also be noticed in Figs.~\ref{fig:sampleF} and \ref{fig:fullF}.

\begin{table}[t]

\centering
\scriptsize
\caption{Skyline explorations of \emph{SkyEx-FES} compared to \emph{SkyEx-F} for pairs that are 30, 50, 80, and 100 m apart}
\begin{tabular}{@{}lllll@{}}
\toprule
\textbf{Distance} & \textbf{30 m} & \textbf{50 m} & \textbf{80 m} & \textbf{100 m} \\ \midrule

\textbf{Number of pairs} & 168193 &  293833 &  557421 & 777452 \\
\textbf{\% of TP} & 14.76\%  & 8.8\% & 4.82\% & 3.49\% \\
\textbf{\emph{SkyEx-F} skylines}          & 1113          & 1182          & 1228          & 1228           \\
\textbf{\emph{SkyEx-FES  }  skylines}      & 403           & 327           & 284           & 274            \\ \bottomrule
\end{tabular}
\label{tab:skyex-fes}
\end{table}

\subsection{SkyEx-D Performance}

In these experiments, we use \emph{SkyEx-D}  (Algorithm \ref{alg:skyex-d}) to set $k_d$ and evaluate our results in terms of \textit{F-measure}. We apply \emph{SkyEx-D} on spatial entities that are 30, 50, 80, and 100 meters apart (see the dataset details in Table \ref{tab:skyex-fes}). We calculate the first derivative ($\mu'_d$) in each point as in Algorithm \ref{alg:skyex-d}. The smoothed $\mu'_d(k)$ with respect to $k$ are presented in Fig. \ref{fig:deriv}. The red solid line shows the value of $k_f$, while the green dashed line represents $k_d$ found by \emph{SkyEx-D}. We note when $\mu'_d(k)$ is negative for the first time and set $k_d$ accordingly.  In the case of spatial entities that are 30 m apart (Fig. \ref{fig:deriv30m}), $k_d$ is only 5 skylines apart from $k_f$ but 73 skylines for 50 m.
These values of $k_d$ are discovered using the first derivative (Eq. 4, Sect. 7.2). We illustrate the trend of $\mu_k$ while increasing $k$, which means that we explore deeper skylines and examine more pairs that are less likely to be a match. The distance from the positive class to the negative is smaller in the beginning because the mean $\mu_k$ is biased by the close points. While we increase $k$, $\mu_k$ increases, meaning that the classes are becoming more and more distinguishable from one another. The high values of  $\mu_k$ show a high distance between the classes. For spatial entities that are 80 m and 100 m apart, $\mu_k$ starts dropping faster than for those that are 30 m and 50 m apart (Fig.~\ref{fig:estimk}). This observation can be justified by the fact that closeby entities are more difficult to classify, so the "grey" area of the potential cut-off is larger. However, \emph{SkyEx-D} detects the first decrease in $\mu_k$ from the first derivative and fixes $k_d$. Graphically, this point coincides with the beginning of the "grey" area. Even though $k_d$  is sometimes fixed far from $k_f$ (m=50), the corresponding \textit{F-measures} are almost the same (Fig.~\ref{fig:f}). 
The red line in Fig.~\ref{fig:f} corresponds to $k_f$ and the green line to $k_d$. The difference in \textit{F-measure} is 0.002 for 30 m, 0.009 for 50 m, 0.002 for 80 m, and 0.004 for 100 m. Thus, \emph{the difference in F-measure for the classifying the pairs using $k_d$ instead of $k_f$ is always less than 0.01.} This means that our \emph{SkyEx-D}, even though fully unsupervised, is almost optimal in terms of \textit{F-measure}.

In terms of precision, recall and \emph{F-measure}, \emph{QuadSky} with \emph{SkyEx} in \cite{Isaj2019multi},  \emph{QuadSky} with \emph{SkyEx-F}, and  \emph{QuadSky} with \emph{SkyEx-FES} report the same values. However, the underlying algorithms are different. \emph{SkyEx} in \cite{Isaj2019multi} needs the threshold $k$ to separate the skylines, whereas for \emph{SkyEx-F} and \emph{SkyEx-FES}, there is no need for specifying $k$ because the algorithms will fix it through the skyline explorations (only 30\% of the skylines for \emph{SkyEx-FES}). \emph{QuadSky} with \emph{SkyEx-D}, being fully unsupervised, might yield different results. The optimal scenario is if it fixes $k_d$ as the $k_f$.

\subsection{Comparison with Baselines}

\begin{table}[b]
\centering
\scriptsize
\caption{Comparison with the baselines}
\begin{tabular}{@{}llllllll@{}}
\toprule
                                                              & \multicolumn{3}{c}{$\boldsymbol{D_{\textit{full}}}$}                     & \multicolumn{3}{c}{$\boldsymbol{D_{\textit{sample}}}$}                  \\ \midrule
\textbf{Approach}                                                              &  \textbf{Prec.} & \textbf{Rec.} & \textbf{F1} &  \textbf{Prec.} & \textbf{Rec.} & \textbf{F1}\\ \midrule
Berjawi et al.(V1)\cite{berjawi2014representing} & \textbf{0.93}        & 0.26  & 0.41 & \textbf{1.00} & 0.27 & 0.43 \\
Berjawi et al.(V1)\cite{berjawi2014representing}-\emph{Flex} & 0.87               & 0.50            & 0.63      & 0.79               & 0.42            & 0.55  
\\ Berjawi et al.(V2)\cite{berjawi2014representing} & 0.73        & 0.56  & 0.63 & 0.97 & 0.60 & 0.74 
\\ Berjawi et al.(V2)\cite{berjawi2014representing}-\emph{Flex}  & 0.73        & 0.56  & 0.63 & 0.82               & 0.76            & 0.79  \\

     Morana et al.\cite{morana2014geobench}  & 0.39      & 0.60   &  0.47 & 0.33 & 0.60 & 0.43 \\
      Karam et al.\cite{karam2010integration}  & 0.23      & \textbf{0.73}   &  0.35 & 0.54 & 0.68 & 0.60 \\
      \emph{QuadSky} with \emph{SkyEx-F}  & 0.87      & 0.60  & \textbf{0.72} & 0.87 & \textbf{0.82} & \textbf{0.85} \\
      \emph{QuadSky} with \emph{SkyEx-D}  & 0.85      & 0.62  & 0.71 & 0.87 & \textbf{0.82} & \textbf{0.85} \\

\bottomrule
\end{tabular}
\label{tab:baselines}
\end{table}

Even though there are several papers in spatial data integration, the works of \cite{berjawi2014representing,morana2014geobench, karam2010integration} are the most similar to ours, as the rest of the related work considers only spatial objects, not spatial entities, or uses supervised learning techniques. We will compare \emph{QuadSky} to Berjawi et al.\cite{berjawi2014representing},  Morana et al.\cite{morana2014geobench}, and Karam et al.\cite{karam2010integration}.
Berjawi et al. \cite{berjawi2014representing} propose Euclidean distance for the geographic coordinates and Levenshtein similarity for all other attributes. The similarities added together to a global similarity. The attributes mentioned in the paper are the name and the phone. However, since the phone is part of our automatic labeling, it can not be used in the algorithm as well. The authors consider pairs with $\textit{score} \geq 0.75$ as a match with high confidence. We use this threshold but also try other thresholds that might yield better results (the versions with the suffix {\emph{-Flex}}). We compare against two versions proposed by the authors: name + address + geographic coordinates (V1) and name + geographic coordinates (V2). 
Morana et al. \cite{morana2014geobench} suggest filtering entities that share the same category or a token in the name. Then these entities are compared using the Euclidean distance for the coordinates, Levenshtein for the address and name, and Resnik similarity (Wordnet) for the category. Attributes like address, phone, etc. are considered secondary, so they are given $\frac{1}{3}$ of the weight in the similarity score function, while name, category, and geographic proximity carry $\frac{2}{3}$ of the weight. The authors show top $k$ matches for each entity to the user to decide.
Karam et al. \cite{karam2010integration} starts with filtering spatial entities that are 5 m apart. Then, the similarity of the name is measured with Levenshtein distance, the geographic similarity with Euclidean distance and the keywords are compared semantically. In order to decide which pairs to match and which not, the similarities are fused using belief theory \cite{raimond2008data}.

 

The results using $D_{\textit{full}}$ and $D_{\textit{sample}}$ are presented in Table~\ref{tab:baselines}. In general, all the methods performed better in $D_{\textit{sample}}$ due to the better quality of the labels. Berjawi et al.(V2)\cite{berjawi2014representing} yields reasonable results, the second best after \emph{QuadSky}, with an \emph{F-measure} of 0.63 in $D_{\textit{full}}$ and 0.74 in $D_{\textit{sample}}$.
If we allow flexible thresholds, Berjawi et al.(V2)\cite{berjawi2014representing}-\emph{Flex} in $D_{\textit{full}}$ finds the same best threshold of 0.75, whereas in $D_{\textit{sample}}$ the threshold of 0.65 yields better results, increasing the  $\textit{F-measure}$ from 0.74 to 0.79 (see Annex D for all the thresholds and their results).
To compare with Morana et al.\cite{morana2014geobench}, we tried all values $k$ from 1 to the maximal matches for a single point (see Fig. 10 in \cite{Isaj2019multi}). 
The highest value of $\textit{F-measure} $ corresponded to a $\textit{p}$ of 0.39 and a $\textit{r}$ of 0.60. The behavior of Morana et al.\cite{morana2014geobench} in $D_{\textit{sample}}$ is similar; the best value of $\textit{F-measure} $ was achieved for $k=3$ and results are similar to those in $D_{\textit{full}}$. The work of Karam et al.\cite{karam2010integration} achieves the highest $\textit{r}$ of 0.73 but a very low value of $\textit{p}$ of 0.23 for $D_{\textit{full}}$. As a result, the $\textit{F-measure}$ is only 0.47. However, in $D_{\textit{sample}}$, the method performs better overall ($\textit{F-measure}$ =0.6). 

\emph{The QuadSky versions  provide the best trade-off between $\textit{p}$ and $\textit{r}$, and thus, the highest $\textit{F-measure}$ in both datasets}. In $D_{\textit{sample}}$, \emph{QuadSky} with \emph{SkyEx-F} and \emph{QuadSky} with \emph{SkyEx-D}  achieve the best $\textit{r}$ compared to all baselines. What is more important, \emph{QuadSky with SkyEx-D, even using an unsupervised algorithm, is still better than the threshold-based baselines.} The highest $\textit{p}$ values for both datasets is achieved by Berjawi et al.(V1)\cite{berjawi2014representing} but a very low $\textit{r}$ and poor model performance overall. In fact, models that achieve extreme values (high precision-low recall or low precision-high recall) are not a viable solution because they are either too restrictive or too flexible, and their predictability is poor. 
Berjawi et al. \cite{berjawi2014representing}(V2)-\emph{Flex} assumes the same weights for all similarities, and the reported values of $\textit{p}$ and $\textit{r}$ are good. However, \emph{the behaviors of the pairs can be of all types. QuadSky can capture these different behaviors better than a simple sum would}.

Regarding the complexity of the baselines, we cannot judge in terms of the blocking techniques because there are no details on whether the authors used an index to create the blocks. However, as we show in Fig. \ref{fig:compQuad}, the available FNN solutions in Postgres still do not scale as well as our \emph{QuadFlex}. Therefore, we perform better in the blocking step. The pairwise comparison has a linear complexity for all baselines and our solution. As for the labeling, the baselines do not need the quadratic complexity induced by our skylines. Our  \emph{SkyEx-*} family of algorithms run for 1 minute in  $D_{\textit{sample}}$ and up to 2 hours in  $D_{\textit{full}}$ with 777,452 pairs. Nevertheless, the entity linkage problem is performed offline, and consequently, even though a fast solution is preferable in general, the effectiveness is much more important, and here \emph{QuadSky} significantly outperforms the baselines.

\subsection{Comparison with Supervised Learning Techniques}

\begin{table}[b]
\centering
\scriptsize
\caption{Comparison with supervised learning}
\begin{tabular}{@{}llllllllll@{}}
\toprule
                         & \multicolumn{3}{c}{\textbf{$\boldsymbol{D_{\textit{full}}}$ -$\boldsymbol{D_{\textit{full}}}$}}      & \multicolumn{3}{c}{\textbf{$\boldsymbol{D_{\textit{sample}}}$-$\boldsymbol{D_{\textit{sample}}}$}}  & \multicolumn{3}{c}{\textbf{$\boldsymbol{D_{\textit{sample}}}$-$\boldsymbol{D_{\textit{full}}}$}}     \\ \midrule
\textbf{Method}          & \textbf{Pr.} & \textbf{Rec.} & \textbf{F1}   & \textbf{Pr.} & \textbf{Rec.} & \textbf{F1}   & \textbf{Pr.} & \textbf{Rec.} & \textbf{F1}   \\ \midrule
{Log. reg.} & 0.83           & 0.70          & \textbf{0.76} & 0.80           & 0.83          & 0.81          & 0.70           & 0.72          & 0.71          \\
{SVM}             & \textbf{0.88}           & 0.67          & \textbf{0.76}          & {0.81}  & 0.80          & 0.81          & 0.71           & 0.70          & 0.71          \\
{Dec. Tree}   & \textbf{0.88}  & 0.66          & 0.75          & \textbf{0.93}           & {0.82} & \textbf{0.87} & 0.65           & 0.74          & 0.69          \\
{Naive B.}      & 0.71           & \textbf{0.77} & 0.74          & 0.63           & \textbf{0.85}          & 0.72          & 0.62           & \textbf{0.77} & 0.69          \\
\emph{SkyEx-F}         & 0.80           & 0.69          & 0.74          & 0.87           & 0.82          & 0.84          & 0.80           & 0.69          & \textbf{0.74} \\
\emph{SkyEx-D}         & 0.81           & 0.68          & 0.74          & 0.87           & 0.82          & 0.84          & \textbf{0.81}  & 0.68          & \textbf{0.74} \\ \bottomrule
\end{tabular}
\label{tab:supervised}
\end{table}

In this section, we keep our \emph{QuadSky} steps but replace the labeling of the pairs with a supervised learning technique. 
We decided to compare the \emph{SkyEx-*} family of algorithms with logistic regression \cite{hosmers}, support vector machines (SVM) \cite{cortes1995support}, decision trees \cite{belson1959matching}, and Naive Bayes \cite{bayes1763lii}, which are supervised learning techniques commonly used in entity resolution problems \cite{isaj2019profile, kopcke2010evaluation, you2011socialsearch, sehgal2006entity, goga2013exploiting}. We applied these methods on $D_{\textit{full}}$ pairs that are at most 30 meters apart (dataset description in Table 3). We experimented with training on 75\% of $D_{\textit{full}}$ and testing on the remaining 25\% with 4-fold cross validation ($D_{\textit{full}}$-$D_{\textit{full}}$), training on 75\% of $D_{\textit{sample}}$ and testing on the remaining 25\% with 4-fold cross validation ($D_{\textit{sample}}$-$D_{\textit{sample}}$), and training on $D_{\textit{sample}}$ and testing on $D_{\textit{full}}$ ($D_{\textit{sample}}$-$D_{\textit{full}}$). 
The results are presented in Table \ref{tab:supervised}. While logistic regression and SVM yield a slightly higher \emph{F-measure} of 0.76 in $D_{\textit{full}}$-$D_{\textit{full}}$, our algorithms, which do not build their model on  labeled data, have almost the same \emph{F-measures} (0.74 for \emph{SkyEx-F} and \emph{SkyEx-D}) in $D_{\textit{full}}$-$D_{\textit{full}}$. 
For the manually labeled dataset in $D_{\textit{sample}}$-$D_{\textit{sample}}$, our algorithms perform the second best (\emph{F-measure} of 0.84), after the decision trees.  \emph{SkyEx-F} and \emph{SkyEx-D} outperform the logistic regression, SVM, and the Naive Bayes, which yield \emph{F-measures} of 0.81, 0.81, and 0.72, respectively.  
Having a large training set as in $D_{\textit{full}}$-$D_{\textit{full}}$ is unrealistic in most real cases. Thus, we tried a more realistic scenario, where one would prepare a small manually labeled training set, and then, test the trained model on the full data ($D_{\textit{sample}}$-$D_{\textit{full}}$). \emph{In this (most realistic) case, {SkyEx-F} and {SkyEx-D} outperform all supervised methods by 0.03-0.05 in {F-measure}, showing the main weakness of supervised models, namely that the $D_{\textit{sample}}$ model is not representative enough when applied to $D_{\textit{full}}$. }

In general, the spatial entity linkage problem suffers from the lack of labeled data \cite{Isaj2019multi,berjawi2014representing,morana2014geobench}. Consequently, the applicability of supervised learning techniques is limited. On the contrary, \emph{SkyEx-D} is completely unsupervised and can still achieve results similar to a supervised technique. 
If the labeled data is present, note that supervised learning techniques build the model on the labeled data, whereas  \emph{SkyEx-F} and \emph{SkyEx-FES} use the labels only to tune the threshold because the construction of the skylines is independent of the labels. For this reason, \emph{in contrast to supervised learning, {SkyEx-F}, and {SkyEx-FES} do not require a big and representative training set, do not struggle with class imbalance, do not overfit the data, and their dimensionality is minimal (one skyline versus high-dimensional data).  }
  

\subsection{Comparison of SkyEx-D to Clustering Techniques}

In Sect. 7.2, we claimed that clustering techniques would not manage to create two clusters: one for the positive-class pairs and one for the negative-class pairs. In this section, we will replace \emph{SkyEx-D} with common clustering techniques and evaluate the formed clusters. We are comparing to distance-based clustering (k-means \cite{macqueen1967some} and k-medoids \cite{kaufman1987clustering}),  hierarchial clustering \cite{hastie2009hierarchical} (agglomerative), and  density-based clustering (DBSCAN \cite{ester1996density}). The results are presented in Table \ref{tab:clustering}. For k-means and k-medoids, we specified the number of clusters as 2. For the hierarchical clustering, we cut the dendrogram to create two clusters. For DBSCAN, we tried several values of minimum points and $\epsilon$ to form either two clusters, or one cluster and noise points. We report the version with the noise points in the table because it yields better results. For the labeling, we tried both versions (labeling cluster 1 as positive and the rest as negative and vice-versa) and report the best version in the table. Distance-based clustering yields the best results, having the highest recall but with a very low precision of 0.28 in $D_{\textit{full}}$, and the second-best (after \emph{SkyEx-D}) \emph{F-measure} of 0.74 in $D_{\textit{sample}}$. 
Hierarchical clustering achieves higher precision than distance-based but with a very low recall of 0.11 in $D_{\textit{full}}$, while the results are reversed to a high recall of 0.91 and a low precision of 0.23 in $D_{\textit{sample}}$. For DBSCAN, the best values were achieved when we labeled the cluster as negative, and the noise points as positive, resulting in a recall of 1.0, but a very low precision of 0.23 in $D_{\textit{full}}$ and 0.26 in $D_{\textit{sample}}$. This means that the positive-class pairs are not dense enough to form a cluster. Our \emph{SkyEx-D} focuses more on the distance between the classes rather than within the classes, and thus outperforms clustering.

\begin{table}[b]
\centering
\scriptsize
\caption{Comparing SkyEx-D to clustering techniques}
\begin{tabular}{@{}lllllll@{}}
\toprule
\textbf{} &  \multicolumn{3}{c}{$\boldsymbol{D_{\textit{full}}}$} &  \multicolumn{3}{c}{$\boldsymbol{D_{\textit{sample}}}$ } \\ \midrule
\textbf{Method} & \textbf{Prec.} & \textbf{Rec.} & \textbf{F1} & \textbf{Prec.} & \textbf{Rec.} & \textbf{F1} \\ \midrule
K-means & 0.28 & 0.96 & 0.44 & 0.62 & 0.92 & 0.74 \\
K-medoids & 0.28 & 0.96 & 0.44 & 0.62 & 0.92 & 0.74 \\
Hierarchial & 0.62 & 0.11 & 0.19 & 0.23 & 0.91 & 0.36 \\
DBSCAN & 0.23 & \textbf{1.00} & 0.37 & 0.26 & \textbf{1.00} & 0.42 \\
\emph{SkyEx-D} & \textbf{0.81} & 0.68 & \textbf{0.74} & \textbf{0.87} & 0.82 & \textbf{0.84} \\ \bottomrule
\end{tabular}
\label{tab:clustering}
\end{table}

\section{Conclusions and Future Work}

Location-based sources provide rich details and semantics about spatial entities. However, identifying which pairs of spatial entities refer to the same physical entity  is a challenging problem. 
In this paper, we addressed the problem of spatial entity linkage across multiple location-based sources. We proposed \emph{QuadSky}, an approach that consists of a spatial blocking technique \emph{QuadFlex}, pairwise comparisons with suitable similarity metrics for each attribute, a skyline-based ranking algorithm \emph{SkyRank}, and the \emph{SkyEx-*} family of algorithms for classifying the pairs. \emph{QuadFlex} arranges the spatial entities into spatial blocks with a low execution time (4-8 times less than FNN \cite{bentley1977complexity}) and without missing relevant comparisons (99.99\% of FNN comparisons). 
\emph{SkyEx-F} achieves 0.84 $\textit{p}$ and 0.84 $\textit{r} $ on a manually labeled dataset and 0.87 $\textit{p}$ and 0.6 $\textit{r}$ on an automatically labeled dataset. We provided a theoretical guarantee to prune 73\% of the skyline explorations in \emph{SkyEx-F} with the novel \emph{SkyEx-FES} without any loss of \emph{F-measure}. Our fully unsupervised \emph{SkyEx-D} finds $k_d$ very close to the optimal $k_f$ (an \emph{F-measure} loss of just 0.01).  The \emph{SkyEx-*} family of algorithms outperforms the existing baselines in terms of $\textit{F-measure}$ and approximates the results of a supervised learning solution without the need of a labeled dataset, while \emph{SkyEx-D} yields far better results than the clustering techniques. \emph{SkyEx-F} and \emph{SkyEx-D} are already available in the R \texttt{skyex} package \cite{isaj2020skyex}, together with other functions for entity linkage. In future work, we aim to study different blocking techniques that combine several attributes and extend our \emph{SkyEx-*} family of algorithms to general (non-spatial) entity resolution problems.



\vspace{-15mm}

\begin{IEEEbiography}[{\includegraphics[width=1in,height=1.25in,clip,keepaspectratio]{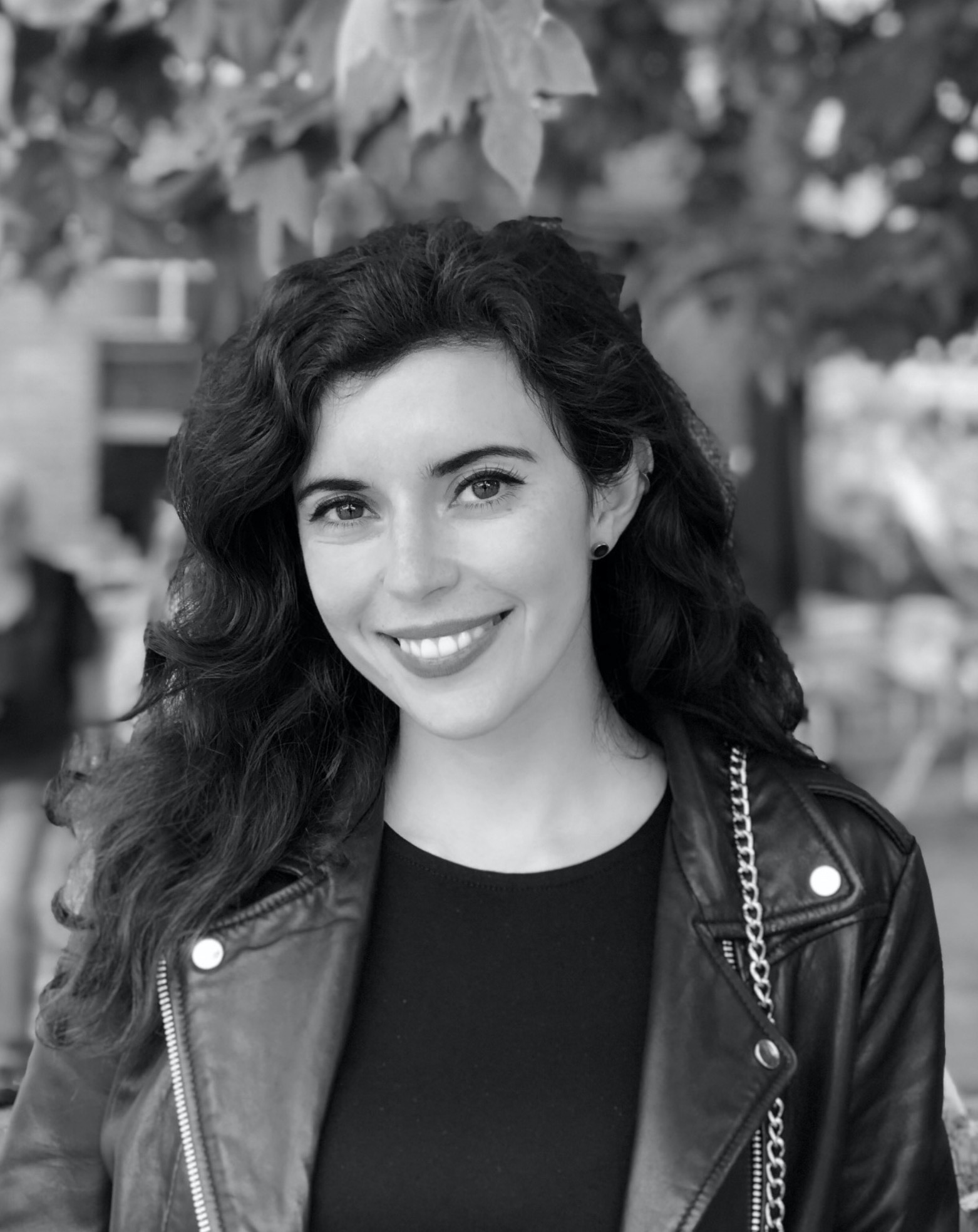}}]{Suela Isaj} is a Ph.D. fellow at the
Center for Data-Intensive Systems (Daisy) at Aalborg University, Denmark. 
She received an M.Sc. degree from the University of Tirana (2013) and a joint M.Sc. degree from  Universit\'{e} libre de Bruxelles, Universit\'{e} Fran\c{c}ois Rabelais and CentraleSup\'{e}lec (2017). 
Her current research is focused on location-based data extraction and integration, data mining, and data analytics. 
\end{IEEEbiography}
\vskip -3\baselineskip plus -1fil

\begin{IEEEbiography}[{\includegraphics[width=1in,height=1.25in,clip,keepaspectratio]{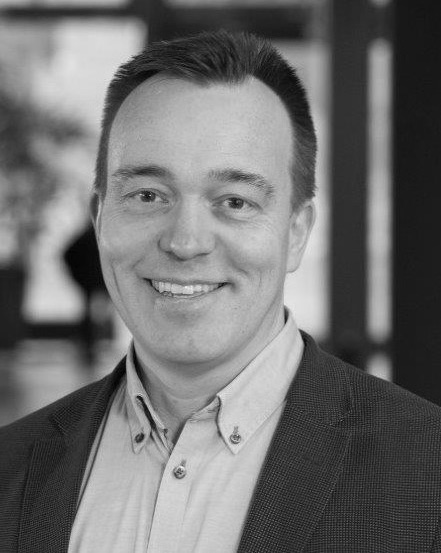}}]
{Torben Bach Pedersen}
is a professor at the
Center for Data-Intensive Systems (Daisy) at
Aalborg University, Denmark. His research concerns
Big Data analytics on multidimensional
Data. He is an ACM Distinguished
scientist, a senior member of the IEEE, and a member of the Danish Academy of Technical Sciences.
\end{IEEEbiography}
\vskip -3\baselineskip plus -1fil
\begin{IEEEbiography}[{\includegraphics[width=1in,height=1.25in,clip,keepaspectratio]{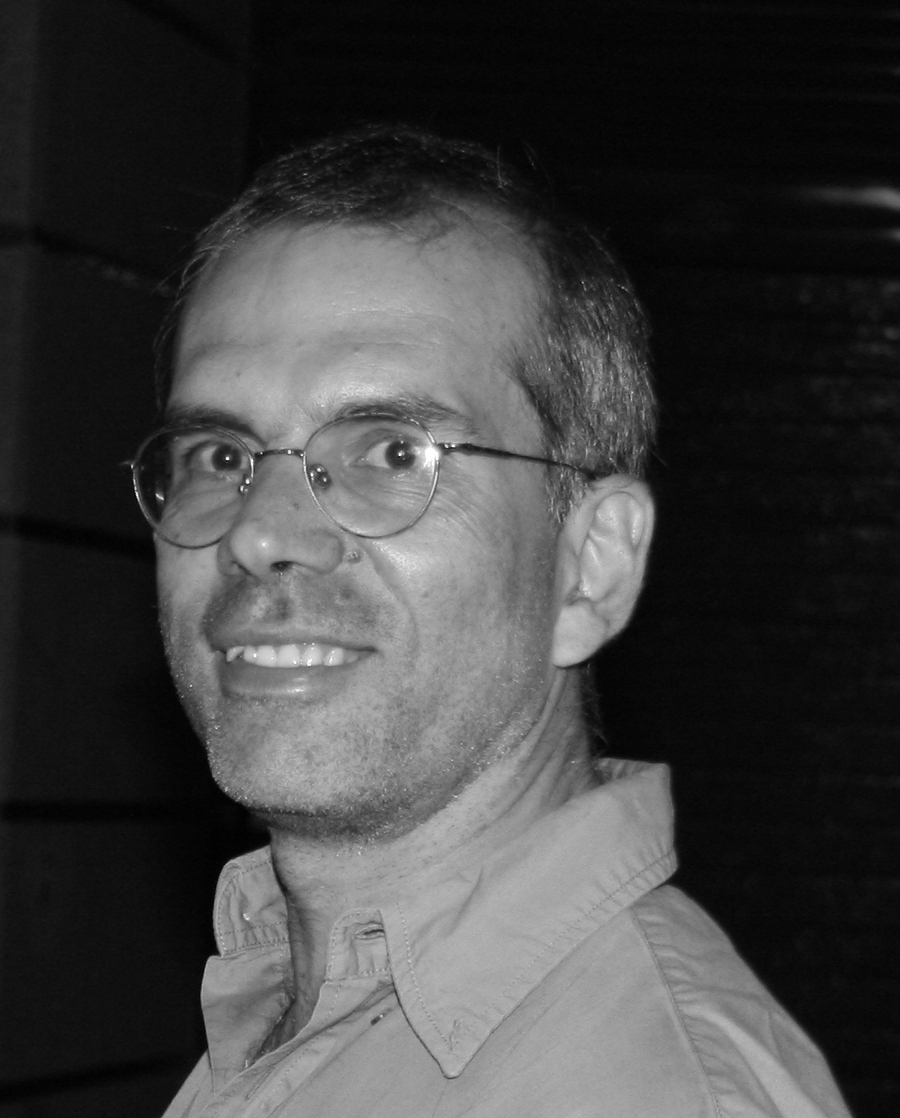}}]{Esteban Zim\'{a}nyi} is a professor and the director of the Department of Computer and Decision Engineering (CoDE) at Universit\'{e} libre de Bruxelles. He is Editor-in-Chief of the Journal on Data Semantics published by Springer. His current research interests include data warehouses, spatio-temporal databases, geographic information systems, and semantic web.
\end{IEEEbiography}

\appendices
\section{Dataset Overview}

In this section, we show the spread of the spatial entities on the map. The map is constructed with layers, starting from the sources with more points, allowing the sources with fewer points to still be visible on the map. The first layer contains the spatial entities from Krak, followed by Google Places, Yelp, and Foursquare spatial entities. Google Places and Krak have the highest overlap, potentially in the southwest area, where Google manages to cover the Krak layer. Foursquare and Yelp have spatial entities are distributed all over the area, sometimes overlapping with each other.

 \begin{figure}[htp]
    \centering
 \includegraphics[width=\linewidth]{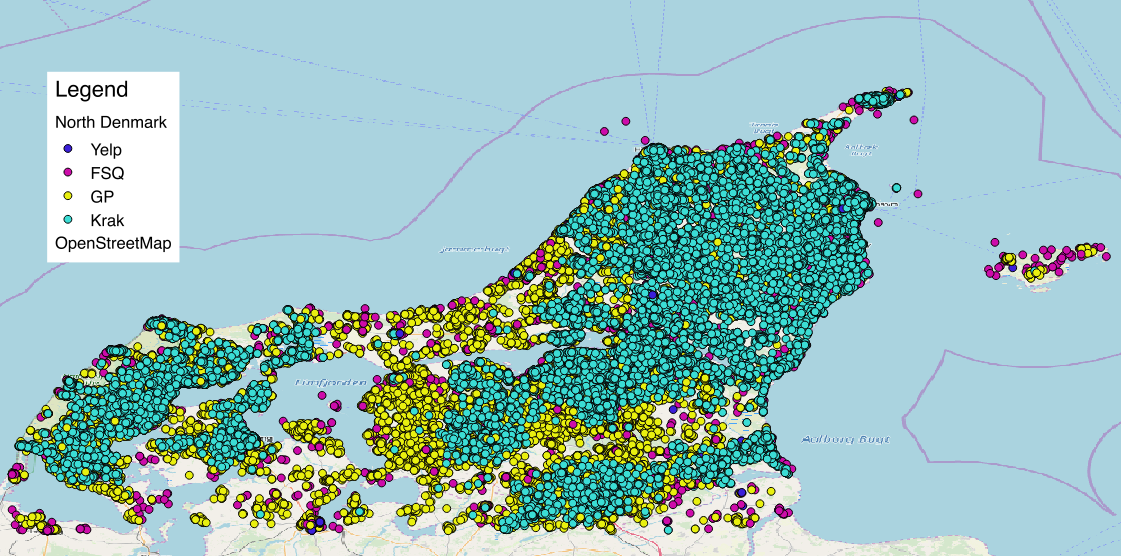}
   \caption{North Denmark dataset}
    \label{fig:Dkall}
\end{figure}  

\section{QuadFlex, quadtree, FNN GiST, FNN SP-GiST, and No-Index}
In this section, we provide further details about the time execution of different spatial blocking solutions: our proposed \emph{QuadFlex}, traditional quadtree, Fixed Radius Nearest Neighbors algorithm \cite{bentley1977complexity} (FNN) implemented in PostgreSQL indexes (GiST  (https://www.postgresql.org/docs/current/gist.html),
and SP-GiST (https://www.postgresql.org/docs/ current/spgist.html), and not using any index (No-Index). Fig. \ref{fig:full} show the scalability of these solutions while increasing the number of points. No-Index has the worst scalability, even for 5000 points, it is 145 times slower than \emph{QuadFlex}, while it becomes 368,095 times slower for 1,000,000 points. Even though the spatial entity linkage is an offline problem, using no index at all is not a feasible solution. 
 
    \begin{figure}[htb]
 \centering
    \includegraphics[width=0.9\linewidth]{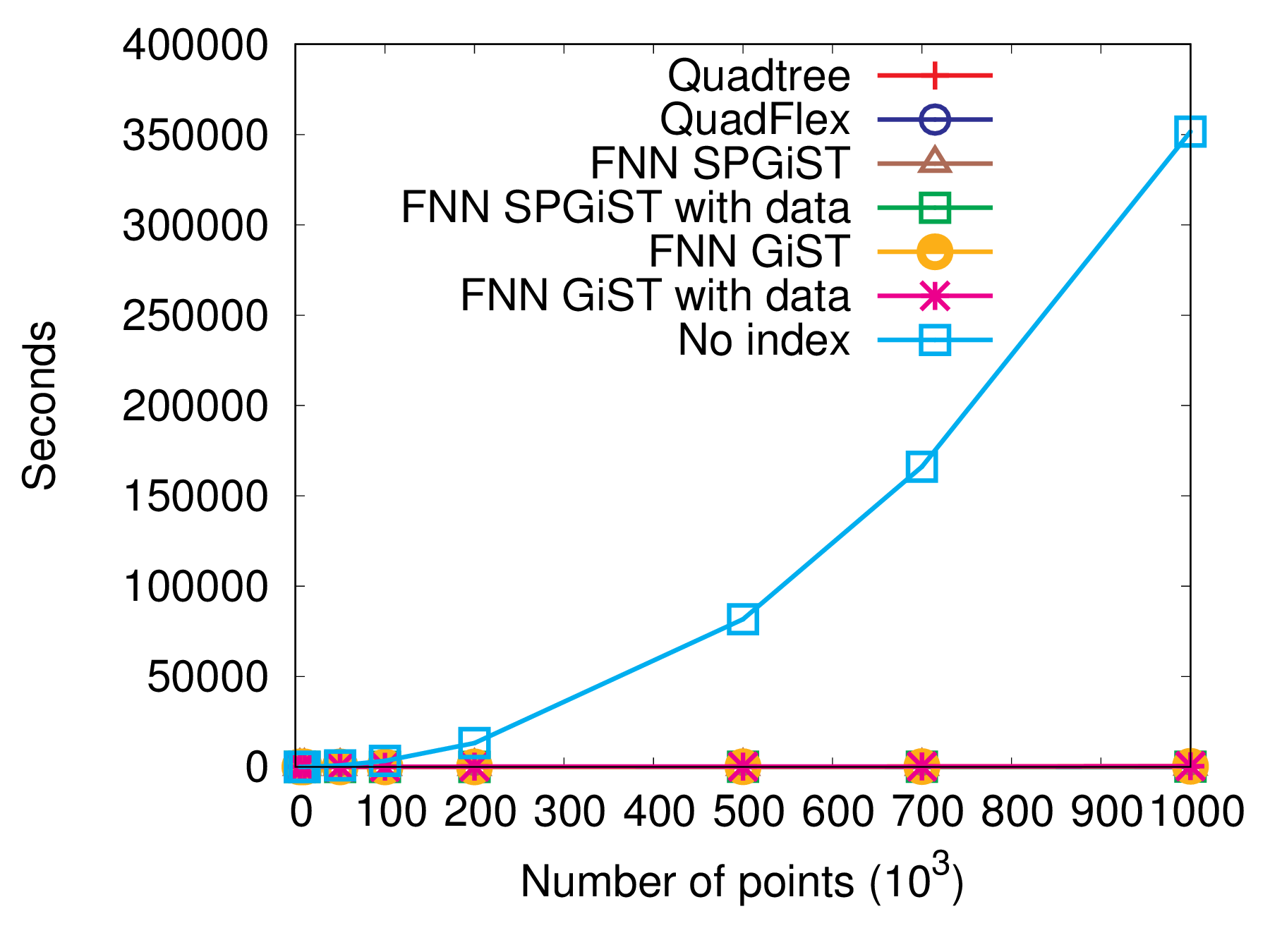}
   \caption{Execution time}
    
 \label{fig:full}
\end{figure}

\section{Experimenting with QuadFlex parameters}
\begin{figure}[htb]
\centering
  \begin{subfigure}[b]{0.45\linewidth}
    \includegraphics[width=\linewidth,trim={0 4mm 0 10mm},clip]{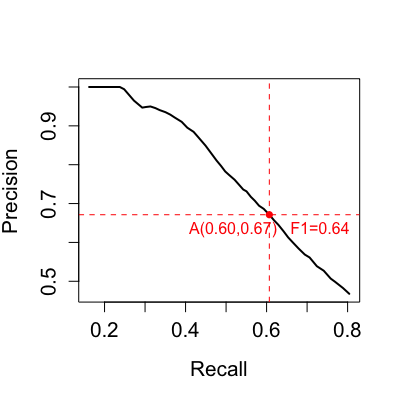}
   \caption{m=1}
    \label{fig:1m}   
  \end{subfigure}\begin{subfigure} [b]{0.45\linewidth}
     \includegraphics[width=\linewidth,trim={0 4mm 0 10mm},clip]{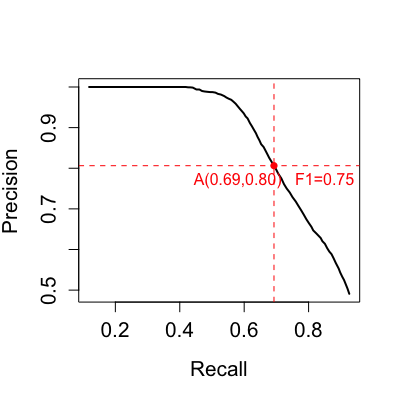}
   \caption{m=20}
    \label{fig:20m}
 \end{subfigure}
 
 \begin{subfigure}[b]{0.45\linewidth}
    \includegraphics[width=\linewidth,trim={0 4mm 0 10mm},clip]{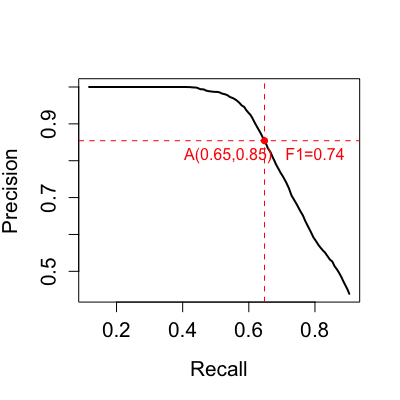}
   \caption{m=40}
    \label{fig:40m}   
  \end{subfigure}\begin{subfigure} [b]{0.45\linewidth}
     \includegraphics[width=\linewidth,trim={0 4mm 0 10mm},clip]{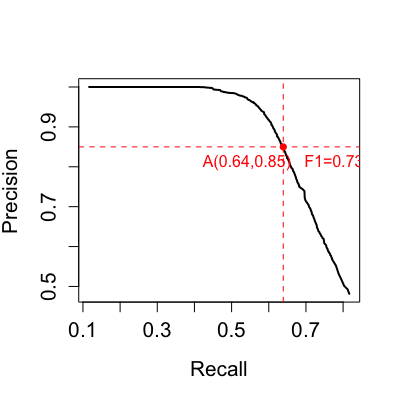}
   \caption{m=60}
    \label{fig:60m}
 \end{subfigure}
 
 \begin{subfigure} [b]{0.45\linewidth}
     \includegraphics[width=\linewidth,trim={0 4mm 0 10mm},clip]{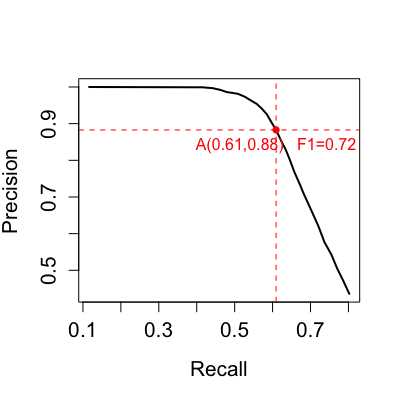}
   \caption{m=80}
    \label{fig:80m}
 \end{subfigure}\begin{subfigure} [b]{0.45\linewidth}
     \includegraphics[width=\linewidth,trim={0 4mm 0 10mm},clip]{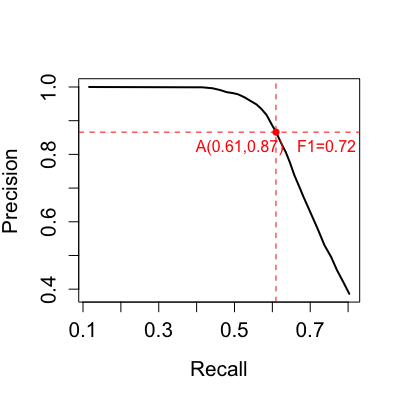}
   \caption{m=100}
    \label{fig:100m}
 \end{subfigure} 
 \caption{Performance of SkyEx-F for different m, no density limit}
 \label{fig:mExp}
\end{figure}

In this section, we experiment with different $m$ and $d$ for the \emph{QuadFlex} and label the pairs with \emph{SkyEx-F}.  We show the precision-recall curves for different $m$ in Fig. \ref{fig:mExp} and for different $d$ in Fig. \ref{fig:dExp}. The values of \emph{F-measure} stay above 0.72 for all the cases, except $m=1$. This result is interesting because the spatial proximity is a good indicator to create blocks, but not necessarily to indicate similarity. Thus, allowing a bigger threshold can be more rewarding. In general, our models yield a very good precision (above 0.8), whereas the recall always stays above 0.6.

 As for the density, all $d$ values yield good precision and recall values. A lower value of $d$ means more splitting of the \emph{QuadFlex} children, e.g., having $d=\frac{10}{1000 m^2}$ in an area is more likely than a $d= \frac{60}{1000 m^2}$. Consequently, the number of nodes that will split further is higher for lower $d$ values. We can notice a small increase in the \emph{F-measure} of 0.02 when using a small $d$ as the areas are better organized, and the entities are compared with more relevant candidates. 

\begin{figure}[htb]
\centering
  \begin{subfigure}[b]{0.45\linewidth}
    \includegraphics[width=\linewidth,trim={0 4mm 0 10mm},clip]{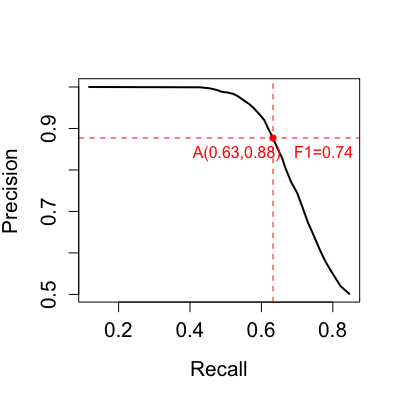}
   \caption{d=10/1000 m\textsuperscript{2}}
    \label{fig:10d}   
  \end{subfigure}\begin{subfigure} [b]{0.45\linewidth}
     \includegraphics[width=\linewidth,trim={0 4mm 0 10mm},clip]{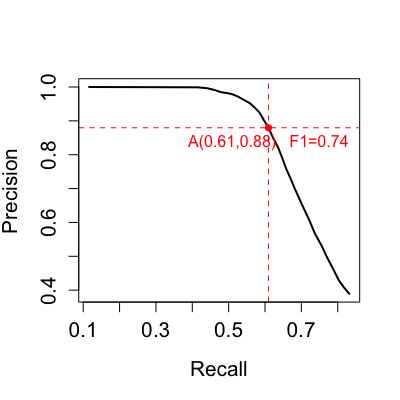}
   \caption{d=20/1000 m\textsuperscript{2}}
    \label{fig:20d}
 \end{subfigure}
 
 \begin{subfigure}[b]{0.45\linewidth}
    \includegraphics[width=\linewidth,trim={0 4mm 0 10mm},clip]{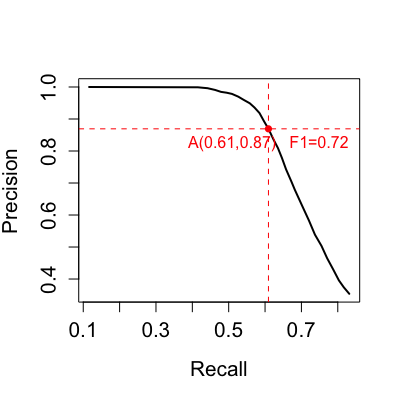}
   \caption{d=30/1000 m\textsuperscript{2}}
    \label{fig:30d}   
  \end{subfigure}\begin{subfigure} [b]{0.45\linewidth}
     \includegraphics[width=\linewidth,trim={0 4mm 0 10mm},clip]{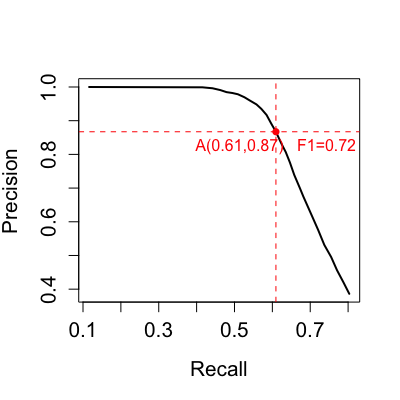}
   \caption{d=40/1000 m\textsuperscript{2}}
    \label{fig:40d}
 \end{subfigure}
 
 \begin{subfigure} [b]{0.45\linewidth}
     \includegraphics[width=\linewidth,trim={0 4mm 0 10mm},clip]{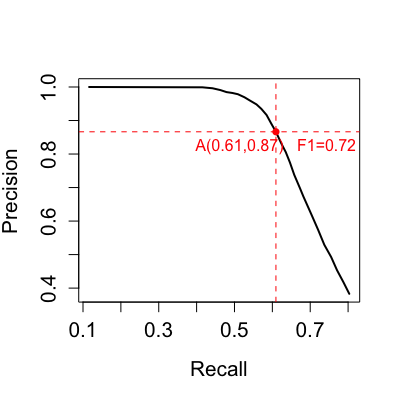}
   \caption{d=50/1000 m\textsuperscript{2}}
    \label{fig:50d}
 \end{subfigure}\begin{subfigure} [b]{0.45\linewidth}
     \includegraphics[width=\linewidth,trim={0 4mm 0 10mm},clip]{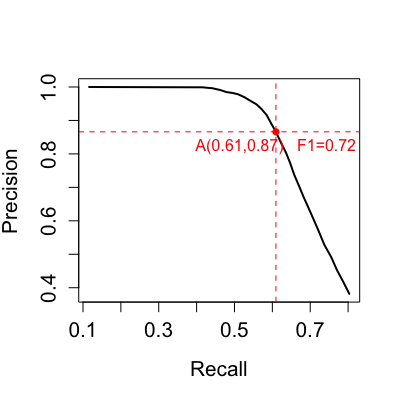}
   \caption{d=60/1000 m\textsuperscript{2}}
    \label{fig:60d}
 \end{subfigure} 
 \caption{Performance of SkyEx-F for different d, m=100}
 \label{fig:dExp}
 
\end{figure}

\section{Comparison to Berjawi et al. with flexible thresholds}
\begin{table}[htb]
\centering
\footnotesize
\caption{Flexible thresholds for Berjawi et al for $D_{full}$}
\begin{tabular}{@{}rrrrrrr@{}}
\toprule
\textbf{}        & \multicolumn{3}{c|}{\textbf{Berjawi et al {[}V1{]-\emph{Flex}}}}                           & \multicolumn{3}{c}{\textbf{Berjawi et al {[}V2{]-\emph{Flex}}}}       \\
\textbf{Cut-off} & \textbf{Prec.} & \textbf{Rec.} & \multicolumn{1}{l|}{\textbf{F1}} & \textbf{Prec.} & \textbf{Rec.} & \textbf{F1} \\ \midrule
0.5              & 0.05               & 0.74            & 0.09                                    & 0.04               & 1.00            & 0.07               \\
0.55             & 0.07               & 0.62            & 0.12                                    & 0.05               & 0.97            & 0.09               \\
0.6              & 0.11               & 0.52            & 0.18                                    & 0.13               & 0.84            & 0.23               \\
0.65             & 0.21               & 0.44            & 0.29                                    & 0.35               & 0.74            & 0.48               \\
0.7              & 0.64               & 0.36            & 0.46                                    & 0.56               & 0.64            & 0.60               \\
0.75             & 0.92               & 0.26            & 0.41                                    & 0.73               & 0.56            & \textbf{0.64}      \\
0.8              & 0.87               & 0.50            & \textbf{0.63}                           & 0.95               & 0.06            & 0.11               \\
0.85             & 0.95               & 0.44            & 0.60                                    & 0.99               & 0.02            & 0.05               \\
0.9              & 0.98               & 0.39            & 0.56                                    & 1.00               & 0.00            & 0.01               \\
0.95             & 0.99               & 0.35            & 0.52                                    & 1.00               & 0.00            & 0.00               \\
1.00                & 1.00               & 0.02            & 0.04                                    & 1.00               & 0.00            & 0.00               \\ \bottomrule
\end{tabular}
\label{tab:thresholds-full}
\end{table}

\begin{table}[htb]

\centering
\footnotesize
\caption{Flexible thresholds for Berjawi et al for $D_{sample}$}
\begin{tabular}{@{}rrrrrrr@{}}
\toprule
\textbf{}        & \multicolumn{3}{c|}{\textbf{Berjawi et al {[}V1{]-\emph{Flex}}}}                           & \multicolumn{3}{c}{\textbf{Berjawi et al {[}V2{]-\emph{Flex}}}}       \\
\textbf{Cut-off} & \textbf{Prec.} & \textbf{Rec.} & \multicolumn{1}{l|}{\textbf{F1}} & \textbf{Prec.} & \textbf{Rec.} & \textbf{F1} \\ \midrule
0.5              & 0.36               & 0.64            & 0.46                                    & 0.25               & 0.86            & 0.38               \\
0.55             & 0.45               & 0.56            & 0.50                                    & 0.29               & 0.84            & 0.44               \\
0.6              & 0.62               & 0.48            & 0.54                                    & 0.58               & 0.80            & 0.67               \\
0.65             & 0.79               & 0.42            & \textbf{0.55}                           & 0.82               & 0.76            & \textbf{0.79}      \\
0.7              & 0.96               & 0.33            & 0.49                                    & 0.94               & 0.67            & 0.78               \\
0.75             & 1.00               & 0.27            & 0.43                                    & 0.97               & 0.60            & 0.74               \\
0.8              & 1.00               & 0.06            & 0.11                                    & 0.99               & 0.55            & 0.71               \\
0.85             & 1.00               & 0.04            & 0.07                                    & 0.99               & 0.49            & 0.65               \\
0.9              & 1.00               & 0.00            & 0.00                                    & 0.99               & 0.44            & 0.60               \\
0.95             & 1.00               & 0.00            & 0.00                                    & 0.99               & 0.39            & 0.56               \\
1.00                & 1.00               & 0.00            & 0.00                                    & 1.00               & 0.03            & 0.06               \\ \bottomrule
\end{tabular}
\label{tab:thresholds-sample}
\end{table}

We provide additional details regarding the comparison to the the baseline \cite{berjawi2014representing} (Sect. 9.7). Berjawi et al. suggest the threshold 0.75 as a cut-off value for the scoring function. However, even though the type of data is similar to ours, there might be a better cut-off value for our dataset. Hence, we experimented with both versions of Berjawi et al.: name + address + geographic coordinates (V1) and name + geographic coordinates (V2).

Tables \ref{tab:thresholds-full} and \ref{tab:thresholds-sample} shows the different thresholds we tried for the work of \cite{berjawi2014representing} on $D_{\textit{full}}$ and $D_{\textit{sample}}$, respectively. The fixed threshold of 0.75 showed to be the best for Berjawi et al(V2)\cite{berjawi2014representing}-\emph{Flex} on $D_{\textit{full}}$. However, the threshold of 0.8 was better for Berjawi et al. (V1)\cite{berjawi2014representing}-\emph{Flex} (F-measure of 0.63 instead of 0.41). Regarding the manually labeled sample,  Table \ref{tab:thresholds-sample} indicates that a threshold of 0.65 would have been more profitable than 0.75, e.g., the \emph{F-measure} would increase from 0.43 to 0.55 for Berjawi et al. (V1)\cite{berjawi2014representing}-\emph{Flex} and from 0.74 to 0.79 for Berjawi et al. (V2)\cite{berjawi2014representing}-\emph{Flex}.

\end{document}